\documentclass[english]{article}
\usepackage[T1]{fontenc}
\usepackage[latin9]{inputenc}
\usepackage{float}
\usepackage{amsthm}
\usepackage{amsmath}
\usepackage{graphicx}
\usepackage{amssymb}
\usepackage{esint}

\makeatletter

\providecommand{\tabularnewline}{\\}
\floatstyle{ruled}
\newfloat{algorithm}{tbp}{loa}
\floatname{algorithm}{Algorithm}

\theoremstyle{plain}
\newtheorem{thm}{Theorem}
  \theoremstyle{plain}
  \newtheorem{lem}[thm]{Lemma}

\usepackage{nips14submit_e,times,sidecap,algorithmic}
\title{Multivariate $f$-Divergence Estimation With Confidence}
\author{Kevin R.~Moon \\ Department of EECS \\ University of Michigan \\ Ann Arbor, MI \\ \texttt{krmoon@umich.edu} \\ \And Alfred O.~Hero III \\ Department of EECS \\ University of Michigan \\ Ann Arbor, MI \\ \texttt{hero@eecs.umich.edu} \\ }

\nipsfinalcopy 

\makeatother

\usepackage{babel}

\begin{document}

\maketitle
\newcommand{\fhat}[1]{\hat{\mathbf{f}}_{#1,k_#1}}
\newcommand{\bz}[1]{b_{#1,k_2}}

\global\long\def\gtay#1#2#3{\mathbf{#1}_{#2}^{(#3)}}
\global\long\def\gk{\hat{\mathbf{G}}_{k_{1},k_{2}}}

\global\long\def\ekl#1{\hat{\mathbf{F}}_{k(#1)}}
\global\long\def\ek{\hat{\mathbf{F}}_{k_{1},k_{2}}}

\global\long\def\ehatl#1#2{\hat{\mathbf{e}}_{#1,k(#2)}}
\global\long\def\ehat#1{\hat{\mathbf{e}}_{#1,k_{#1}}}

\global\long\def\fhat#1{\hat{\mathbf{f}}_{#1,k_{#1}}}
\global\long\def\fhatl#1#2{\hat{\mathbf{f}}_{#1,k(#2)}}
\global\long\def\ft#1#2{\tilde{\mathbf{f}}_{#1,k(#2)}}

\global\long\def\lhat{\hat{\mathbf{L}}_{k_{1},k_{2}}}
\global\long\def\lhatl#1{\hat{\mathbf{L}}_{k(#1)}}

\global\long\def\bz#1#2{b_{#1,k(#2)}}

\global\long\def\bE{\mathbb{E}}
\global\long\def\ez{\mathbb{E}_{\mathbf{Z}}}
\global\long\def\var{\text{Var}}
\global\long\def\bias{\text{Bias}}

\begin{abstract}
The problem of $f$-divergence estimation is important in the fields
of machine learning, information theory, and statistics. While several
nonparametric divergence estimators exist, relatively few have known
convergence properties. In particular, even for those estimators whose
MSE convergence rates are known, the asymptotic distributions are
unknown. We establish the asymptotic normality of a recently proposed
ensemble estimator of $f$-divergence between two distributions from
a finite number of samples. This estimator has MSE convergence rate
of $O\left(\frac{1}{T}\right)$, is simple to implement, and performs
well in high dimensions. This theory enables us to perform divergence-based
inference tasks such as testing equality of pairs of distributions
based on empirical samples. We experimentally validate our theoretical
results and, as an illustration, use them to empirically bound the
best achievable classification error.
\end{abstract}

\section{Introduction}

This paper establishes the asymptotic normality of a nonparametric
estimator of the $f$-divergence between two distributions from a
finite number of samples. For many nonparametric divergence estimators
the large sample consistency has already been established and the
mean squared error (MSE) convergence rates are known for some. However,
there are few results on the asymptotic distribution of non-parametric
divergence estimators. Here we show that the asymptotic distribution
is Gaussian for the class of ensemble $f$-divergence estimators~\cite{moon2014isit},
extending theory for entropy estimation~\cite{sricharan2012nips,sricharan2013ensemble}
to divergence estimation. $f$-divergence is a measure of the difference
between distributions and is important to the fields of machine learning,
information theory, and statistics~\cite{csiszar1967information}.
The $f$-divergence generalizes several measures including the Kullback-Leibler
(KL)~\cite{kullback1951divergence} and R\'{e}nyi-$\alpha$~\cite{renyi1961divergence}
divergences. Divergence estimation is useful for empirically estimating
the decay rates of error probabilities of hypothesis testing~\cite{cover2006infotheory},
extending machine learning algorithms to distributional features~\cite{poczos2011divergence,oliva2013distribution},
and other applications such as text/multimedia clustering~\cite{dhillon2003divisive}.
Additionally, a special case of the KL divergence is mutual information
which gives the capacities in data compression and channel coding~\cite{cover2006infotheory}.
Mutual information estimation has also been used in machine learning
applications such as feature selection~\cite{peng2005feature}, fMRI
data processing~\cite{chai2009fmri}, clustering~\cite{lewi2006cluster},
and neuron classification~\cite{schneidman2002information}. Entropy
is also a special case of divergence where one of the distributions
is the uniform distribution. Entropy estimation is useful for intrinsic
dimension estimation~\cite{carter2010dim}, texture classification
and image registration~\cite{hero2002graphs}, and many other applications.

However, one must go beyond entropy and divergence estimation in order
to perform inference tasks on the divergence. An example of an inference
task is detection: to test the null hypothesis that the divergence
is zero, i.e., testing that the two populations have identical distributions.
Prescribing a p-value on the null hypothesis requires specifying the
null distribution of the divergence estimator. Another statistical
inference problem is to construct a confidence interval on the divergence
based on the divergence estimator. This paper provides solutions to
these inference problems by establishing large sample asymptotics
on the distribution of divergence estimators. In particular we consider
the asymptotic distribution of the nonparametric weighted ensemble
estimator of $f$-divergence from~\cite{moon2014isit}. This estimator
estimates the $f$-divergence from two finite populations of i.i.d.
samples drawn from some unknown, nonparametric, smooth, $d$-dimensional
distributions. The estimator~\cite{moon2014isit} achieves a MSE
convergence rate of $O\left(\frac{1}{T}\right)$ where $T$ is the
sample size. See~\cite{arxiv2014div} for proof details.

\subsection{Related Work}

Estimators for some $f$-divergences already exist. For example, P\'{o}czos
\& Schneider~\cite{poczos2011divergence} and Wang et al~\cite{wang2009divergence}
provided consistent $k$-nn estimators for R\'{e}nyi-$\alpha$ and
the KL divergences, respectively. Consistency has been proven for
other mutual information and divergence estimators based on plug-in
histogram schemes~\cite{darbellay1999MIest,wang2005divergencepart,silva2010partition,le2013partition}.
Hero et al~\cite{hero2002graphs} provided an estimator for R\'{e}nyi-$\alpha$
divergence but assumed that one of the densities was known. However
none of these works study the convergence rates of their estimators
nor do they derive the asymptotic distributions.

Recent work has focused on deriving convergence rates for divergence
estimators. Nguyen et al~\cite{nguyen2010divergence}, Singh and
P\'{o}czos~\cite{singh2014exp}, and Krishnamurthy et al~\cite{krishnamurthy2014div}
each proposed divergence estimators that achieve the parametric convergence
rate ($O\left(\frac{1}{T}\right)$) under weaker conditions than those
given in~\cite{moon2014isit}. However, solving the convex problem
of~\cite{nguyen2010divergence} can be more demanding for large sample
sizes than the estimator given in~\cite{moon2014isit} which depends
only on simple density plug-in estimates and an offline convex optimization
problem. Singh and P\'{o}czos only provide an estimator for R\'{e}nyi-$\alpha$
divergences that requires several computations at each boundary of
the support of the densities which becomes difficult to implement
as $d$ gets large. Also, this method requires knowledge of the support
of the densities which may not be possible for some problems. In contrast,
while the convergence results of the estimator in~\cite{moon2014isit}
requires the support to be bounded, knowledge of the support is not
required for implementation. Finally, the estimators given in~\cite{krishnamurthy2014div}
estimate divergences that include functionals of the form $\int f_{1}^{\alpha}(x)f_{2}^{\beta}(x)d\mu(x)$
for given $\alpha,\,\beta$. While a suitable $\alpha$-$\beta$ indexed
sequence of divergence functionals of the form in~\cite{krishnamurthy2014div}
can be made to converge to the KL divergence, this does not guarantee
convergence of the corresponding sequence of divergence estimates,
whereas the estimator in~\cite{moon2014isit} can be used to estimate
the KL divergence. Also, for some divergences of the specified form,
numerical integration is required for the estimators in~\cite{krishnamurthy2014div},
which can be computationally difficult. In any case, the asymptotic
distributions of the estimators in~\cite{nguyen2010divergence,singh2014exp,krishnamurthy2014div}
are currently unknown.

Asymptotic normality has been established for certain appropriately
normalized divergences between a specific density estimator and the
true density~\cite{berlinet1995l1,berlinet1997asymptotic,bickel1973some}.
However, this differs from our setting where we assume that both densities
are unknown. Under the assumption that the two densities are smooth,
lower bounded, and have bounded support, we show that an appropriately
normalized weighted ensemble average of kernel density plug-in estimators
of $f$-divergence converges in distribution to the standard normal
distribution. This is accomplished by constructing a sequence of interchangeable
random variables and then showing (by concentration inequalities and
Taylor series expansions) that the random variables and their squares
are asymptotically uncorrelated. The theory developed to accomplish
this can also be used to derive a central limit theorem for a weighted
ensemble estimator of entropy such as the one given in~\cite{sricharan2013ensemble}.We
verify the theory by simulation. We then apply the theory to the practical
problem of empirically bounding the Bayes classification error probability
between two population distributions, without having to construct
estimates for these distributions or implement the Bayes classifier.

Bold face type is used in this paper for random variables and random
vectors. Let $f_{1}$ and $f_{2}$ be densities and define $L(x)=\frac{f_{1}(x)}{f_{2}(x)}$.
The conditional expectation given a random variable $\mathbf{Z}$
is $\mathbb{E}_{\mathbf{Z}}$.

\section{The Divergence Estimator\label{sec:Weighted_ensemble}}

Moon and Hero~\cite{moon2014isit} focused on estimating divergences
that include the form~\cite{csiszar1967information} \begin{equation}
G(f_{1},f_{2})=\int g\left(\frac{f_{1}(x)}{f_{2}(x)}\right)f_{2}(x)dx,\label{eq:fdivergences}\end{equation}
for a smooth, function $g(f)$. (Note that although $g$ must be convex
for (\ref{eq:fdivergences}) to be a divergence, the estimator in~\cite{moon2014isit}
does not require convexity.) The divergence estimator is constructed
using $k$-nn density estimators as follows. Assume that the $d$-dimensional
multivariate densities $f_{1}$ and $f_{2}$ have finite support $\mathcal{S}=\left[a,b\right]^{d}$.
Assume that $T=N+M_{2}$ i.i.d. realizations $\left\{ \mathbf{X}_{1},\dots,\mathbf{X}_{N},\mathbf{X}_{N+1},\dots,\mathbf{X}_{N+M_{2}}\right\} $
are available from the density $f_{2}$ and $M_{1}$ i.i.d. realizations
$\left\{ \mathbf{Y}_{1},\dots,\mathbf{Y}_{M_{1}}\right\} $ are available
from the density $f_{1}$. Assume that $k_{i}\leq M_{i}.$ Let $\rho_{2,k_{2}}(i)$
be the distance of the $k_{2}$th nearest neighbor of $\mathbf{X}_{i}$
in $\left\{ \mathbf{X}_{N+1},\dots,\mathbf{X}_{T}\right\} $ and let
$\mathbf{\rho}_{1,k_{1}}(i)$ be the distance of the $k_{1}$th nearest
neighbor of $\mathbf{X}_{i}$ in $\left\{ \mathbf{Y}_{1},\dots,\mathbf{Y}_{M_{1}}\right\} .$
Then the $k$-nn density estimate is~\cite{loftsgaarden1965knn}\[
\fhat i(X_{j})=\frac{k_{i}}{M_{i}\bar{c}\mathbf{\mathbf{\rho}}_{i,k_{i}}^{d}(j)},\]
where $\bar{c}$ is the volume of a $d$-dimensional unit ball.

To construct the plug-in divergence estimator, the data from $f_{2}$
are randomly divided into two parts $\left\{ \mathbf{X}_{1},\dots,\mathbf{X}_{N}\right\} $
and $\left\{ \mathbf{X}_{N+1},\dots,\mathbf{X}_{N+M_{2}}\right\} $.
The $k$-nn density estimate $\fhat 2$ is calculated at the $N$
points $\left\{ \mathbf{X}_{1},\dots,\mathbf{X}_{N}\right\} $ using
the $M_{2}$ realizations $\left\{ \mathbf{X}_{N+1},\dots,\mathbf{X}_{N+M_{2}}\right\} $.
Similarly, the $k$-nn density estimate $\fhat 1$ is calculated at
the $N$ points $\left\{ \mathbf{X}_{1},\dots,\mathbf{X}_{N}\right\} $
using the $M_{1}$ realizations $\left\{ \mathbf{Y}_{1},\dots,\mathbf{Y}_{M_{1}}\right\} $.
Define $\lhat(x)=\frac{\fhat 1(x)}{\fhat 2(x)}.$ The functional $G(f_{1},f_{2})$
is then approximated as \begin{equation}
\gk=\frac{1}{N}\sum_{i=1}^{N}g\left(\lhat\left(\mathbf{X}_{i}\right)\right).\label{eq:estimator}\end{equation}

The principal assumptions on the densities $f_{1}$ and $f_{2}$ and
the functional $g$ are that: 1) $f_{1}$, $f_{2},$ and $g$ are
smooth; 2) $f_{1}$ and $f_{2}$ have common bounded support sets
$\mathcal{S}$; 3) $f_{1}$ and $f_{2}$ are strictly lower bounded.
The full assumptions $(\mathcal{A}.0)-(\mathcal{A}.5)$ are given
in the appendices and in\cite{arxiv2014div}. Moon and Hero~\cite{moon2014isit}
showed that under these assumptions, the MSE convergence rate of the
estimator in Eq.~\ref{eq:estimator} to the quantity in Eq.~\ref{eq:fdivergences}
depends exponentially on the dimension $d$ of the densities. However,
Moon and Hero also showed that an estimator with the parametric convergence
rate $O(1/T)$ can be derived by applying the theory of optimally
weighted ensemble estimation as follows.

Let $\bar{l}=\left\{ l_{1},\dots,l_{L}\right\} $ be a set of index
values and $T$ the number of samples available. For an indexed ensemble
of estimators $\left\{ \hat{\mathbf{E}}_{l}\right\} _{l\in\bar{l}}$
of the parameter $E$, the weighted ensemble estimator with weights
$w=\left\{ w\left(l_{1}\right),\dots,w\left(l_{L}\right)\right\} $
satisfying $\sum_{l\in\bar{l}}w(l)=1$ is defined as $\hat{\mathbf{E}}_{w}=\sum_{l\in\bar{l}}w\left(l\right)\hat{\mathbf{E}}_{l}.$
The key idea to reducing MSE is that by choosing appropriate weights
$w$, we can greatly decrease the bias in exchange for some increase
in variance. Consider the following conditions on $\left\{ \hat{\mathbf{E}}_{l}\right\} _{l\in\bar{l}}$~\cite{sricharan2013ensemble}:
\begin{itemize}
\item $\mathcal{C}.1$ The bias is given by \[
\bias\left(\hat{\mathbf{E}}_{l}\right)=\sum_{i\in J}c_{i}\psi_{i}(l)T^{-i/2d}+O\left(\frac{1}{\sqrt{T}}\right),\]
 where $c_{i}$ are constants depending on the underlying density,
$J=\left\{ i_{1},\dots,i_{I}\right\} $ is a finite index set with
$I<L$, $\min(J)>0$ and $\max(J)\leq d$, and $\psi_{i}(l)$ are
basis functions depending only on the parameter $l$. 
\item $\mathcal{C}.2$ The variance is given by \[
\var\left[\hat{\mathbf{E}}_{l}\right]=c_{v}\left(\frac{1}{T}\right)+o\left(\frac{1}{T}\right).\]
\end{itemize}
\begin{thm}
\cite{sricharan2013ensemble} \label{thm:ensemble}Assume conditions
$\mathcal{C}.1$ and $\mathcal{C}.2$ hold for an ensemble of estimators
$\left\{ \hat{\mathbf{E}}_{l}\right\} _{l\in\bar{l}}$. Then there
exists a weight vector $w_{0}$ such that \[
\mathbb{E}\left[\left(\hat{\mathbf{E}}_{w_{0}}-E\right)^{2}\right]=O\left(\frac{1}{T}\right).\]
The weight vector $w_{0}$ is the solution to the following convex
optimization problem:\[
\begin{array}{rl}
\min_{w} & ||w||_{2}\\
subject\, to & \sum_{l\in\bar{l}}w(l)=1,\\
 & \gamma_{w}(i)=\sum_{l\in\bar{l}}w(l)\psi_{i}(l)=0,\, i\in J.\end{array}\]

\end{thm}
In order to achieve the rate of $O\left(1/T\right)$ it is not necessary
for the weights to zero out the lower order bias terms, i.e. that
$\gamma_{w}(i)=0,\, i\in J$. It was shown in~\cite{sricharan2013ensemble}
that solving the following convex optimization problem in place of
the optimization problem in Theorem~\ref{thm:ensemble} retains the
MSE convergence rate of $O\left(1/T\right)$:\begin{equation}
\begin{array}{rl}
\min_{w} & \epsilon\\
subject\, to & \sum_{l\in\bar{l}}w(l)=1,\\
 & \left|\gamma_{w}(i)T^{\frac{1}{2}-\frac{i}{2d}}\right|\leq\epsilon,\,\, i\in J,\\
 & \left\Vert w\right\Vert _{2}^{2}\leq\eta,\end{array}\label{eq:opt_prob2}\end{equation}
where the parameter $\eta$ is chosen to trade-off between bias and
variance. Instead of forcing $\gamma_{w}(i)=0,$ the relaxed optimization
problem uses the weights to decrease the bias terms at the rate of
$O(1/\sqrt{T})$ which gives an MSE rate of $O(1/T)$.

Theorem~\ref{thm:ensemble} was applied in~\cite{sricharan2013ensemble}
to obtain an entropy estimator with convergence rate $O\left(1/T\right).$
Moon and Hero~ \cite{moon2014isit} similarly applied Theorem~\ref{thm:ensemble}
to obtain a divergence estimator with the same rate in the following
manner. Let $L>I=d-1$ and choose $\bar{l}=\left\{ l_{1},\dots,l_{L}\right\} $
to be positive real numbers. Assume that $M_{1}=O\left(M_{2}\right).$
Let $k(l)=l\sqrt{M_{2}}$, $M_{2}=\alpha T$ with $0<\alpha<1$, $\hat{\mathbf{G}}_{k(l)}:=\hat{\mathbf{G}}_{k(l),k(l)},$
and $\hat{\mathbf{G}}_{w}:=\sum_{l\in\bar{l}}w(l)\hat{\mathbf{G}}_{k(l)}.$
Note that the parameter $l$ indexes over different neighborhood sizes
for the $k$-nn density estimates. From~\cite{moon2014isit}, the
biases of the ensemble estimators $\left\{ \hat{\mathbf{G}}_{k(l)}\right\} _{l\in\bar{l}}$
satisfy the condition $\mathcal{C}.1$ when $\psi_{i}(l)=l^{i/d}$
and $J=\{1,\dots,d-1\}$. The general form of the variance of $\hat{\mathbf{G}}_{k(l)}$
also follows $\mathcal{C}.2$. The optimal weight $w_{0}$ is found
by using Theorem~\ref{thm:ensemble} to obtain a plug-in $f$-divergence
estimator with convergence rate of $O\left(1/T\right).$ The estimator
is summarized in Algorithm~\ref{alg:estimator}.

\begin{algorithm}
\begin{algorithmic}[1]
\renewcommand{\algorithmicrequire}{\textbf{Input:}} \renewcommand{\algorithmicensure}{\textbf{Output:}}

\REQUIRE $\alpha$, $\eta$, $L$ positive real numbers $\bar{l}$,
samples $\left\{ \mathbf{Y}_{1},\dots,\mathbf{Y}_{M_{1}}\right\} $
from $f_{1}$, samples $\left\{ \mathbf{X}_{1},\dots,\mathbf{X}_{T}\right\} $
from $f_{2}$, dimension $d$, function $g$, $\bar{c}$

\ENSURE The optimally weighted divergence estimator $\hat{\mathbf{G}}_{w_{0}}$

\STATE Solve for $w_{0}$ using Eq.~\ref{eq:opt_prob2} with basis
functions $\psi_{i}(l)=l^{i/d}$, $l\in\bar{l}$ and $i\in\{1,\dots,d-1\}$

\STATE $M_{2}\leftarrow\alpha T$, $N\leftarrow T-M_{2}$

\FORALL{$l\in\bar{l}$}

\STATE $k(l)\leftarrow l\sqrt{M_{2}}$

\FOR{$i=1$ to $N$}

\STATE $\rho_{j,k(l)}(i)\leftarrow$the distance of the $k(l)$th
nearest neighbor of $\mathbf{X}_{i}$ in $\left\{ \mathbf{Y}_{1},\dots,\mathbf{Y}_{M_{1}}\right\} $
and $\left\{ \mathbf{X}_{N+1},\dots,\mathbf{X}_{T}\right\} $ for
$j=1,2$, respectively

\STATE $\fhatl jl(\mathbf{X}_{i})\leftarrow\frac{k(l)}{M_{j}\bar{c}\mathbf{\mathbf{\rho}}_{j,k(l)}^{d}(i)}$
for $j=1,2$, $\lhatl l(\mathbf{X}_{i})\leftarrow\frac{\fhatl 1l}{\fhatl 2l}$

\ENDFOR 

\STATE $\hat{\mathbf{G}}_{k(l)}\leftarrow\frac{1}{N}\sum_{i=1}^{N}g\left(\lhatl l(\mathbf{X}_{i})\right)$

\ENDFOR 

\STATE $\hat{\mathbf{G}}_{w_{0}}\leftarrow\sum_{l\in\bar{l}}w_{0}(l)\hat{\mathbf{G}}_{k(l)}$

\end{algorithmic}

\caption{Optimally weighted ensemble divergence estimator\label{alg:estimator}}

\end{algorithm}

\section{Asymptotic Normality of the Estimator}

The following theorem shows that the appropriately normalized ensemble
estimator $\hat{\mathbf{G}}_{w}$ converges in distribution to a normal
random variable.
\begin{thm}
\label{thm:clt}Assume that assumptions $(\mathcal{A}.0)-(\mathcal{A}.5)$
hold and let $M=O(M_{1})=O(M_{2})$ and $k(l)=l\sqrt{M}$ with $l\in\bar{l}$.
The asymptotic distribution of the weighted ensemble estimator $\hat{\mathbf{G}}_{w}$
is given by \[
\lim_{M,N\rightarrow\infty}Pr\left(\frac{\hat{\mathbf{G}}_{w}-\bE\left[\hat{\mathbf{G}}_{w}\right]}{\sqrt{\var\left[\hat{\mathbf{G}}_{w}\right]}}\leq t\right)=Pr(\mathbf{S}\leq t),\]
where $\mathbf{S}$ is a standard normal random variable. Also $\bE\left[\hat{\mathbf{G}}_{w}\right]\rightarrow G(f_{1},f_{2})$
and $\var\left[\hat{\mathbf{G}}_{w}\right]\rightarrow0$.
\end{thm}
The results on the mean and variance come from~\cite{moon2014isit}.
The proof of the distributional convergence is outlined below and
is based on constructing a sequence of interchangeable random variables
$\left\{ \mathbf{Y}_{M,i}\right\} _{i=1}^{N}$ with zero mean and
unit variance. We then show that the $\mathbf{Y}_{M,i}$ are asymptotically
uncorrelated and that the $\mathbf{Y}_{M,i}^{2}$ are asymptotically
uncorrelated as $M\rightarrow\infty$. This is similar to what was
done in~\cite{sricharan2012estimation} to prove a central limit
theorem for a density plug-in estimator of entropy. Our analysis for
the ensemble estimator of divergence is more complicated since we
are dealing with a functional of two densities and a weighted ensemble
of estimators. In fact, some of the equations we use to prove Theorem~\ref{thm:clt}
can be used to prove a central limit theorem for a weighted ensemble
of entropy estimators such as that given in~\cite{sricharan2013ensemble}.

\subsection{Proof Sketch of Theorem~\ref{thm:clt}}

The full proof is included in the appendices. We use the following
lemma from~\cite{sricharan2012estimation,kumar2012thesis}:
\begin{lem}
\label{lem:clt_covariance}Let the random variables $\{\mathbf{Y}_{M,i}\}_{i=1}^{N}$
belong to a zero mean, unit variance, interchangeable process for
all values of $M$. Assume that $Cov(\mathbf{Y}_{M,1},\mathbf{Y}_{M,2})$
and $Cov(\mathbf{Y}_{M,1}^{2},\mathbf{Y}_{M,2}^{2})$ are $O(1/M)$.
Then the random variable \begin{equation}
\mathbf{S}_{N,M}=\left(\sum_{i=1}^{N}\mathbf{Y}_{M,i}\right)/\sqrt{\var\left[\sum_{i=1}^{N}\mathbf{Y}_{M,i}\right]}\label{eq:sum}\end{equation}
converges in distribution to a standard normal random variable.
\end{lem}
This lemma is an extension of work by Blum et al~\cite{blum1958central}
which showed that if $\left\{ \mathbf{Z}_{i};i=1,2,\dots\right\} $
is an interchangeable process with zero mean and unit variance, then
$\mathbf{S}_{N}=\frac{1}{\sqrt{N}}\sum_{i=1}^{N}\mathbf{Z}_{i}$ converges
in distribution to a standard normal random variable if and only if
$Cov\left[\mathbf{Z}_{1},\mathbf{Z}_{2}\right]=0$ and $Cov\left[\mathbf{Z}_{1}^{2},\mathbf{Z}_{2}^{2}\right]=0$.
In other words, the central limit theorem holds if and only if the
interchangeable process is uncorrelated and the squares are uncorrelated.
Lemma~\ref{lem:clt_covariance} shows that for a correlated interchangeable
process, a sufficient condition for a central limit theorem is for
the interchangeable process and the squared process to be asymptotically
uncorrelated with rate $O(1/M)$.

For simplicity, let $M_{1}=M_{2}=M$ and $\lhatl l:=\hat{\mathbf{L}}_{k(l),k(l)}$.
Define \[
\mathbf{Y}_{M,i}=\frac{\sum_{l\in\bar{l}}w(l)g\left(\lhatl l(\mathbf{X}_{i})\right)-\bE\left[\sum_{l\in\bar{l}}w(l)g\left(\lhatl l(\mathbf{X}_{i})\right)\right]}{\sqrt{\var\left[\sum_{l\in\bar{l}}w(l)g\left(\lhatl l(\mathbf{X}_{i})\right)\right]}}.\]
 Then from Eq.~\ref{eq:sum}, we have that \[
\mathbf{S}_{N,M}=\left(\hat{\mathbf{G}}_{w}-\bE\left[\hat{\mathbf{G}}_{w}\right]\right)/\sqrt{\var\left[\hat{\mathbf{G}}_{w}\right]}.\]
 Thus it is sufficient to show from Lemma~\ref{lem:clt_covariance}
that $Cov(\mathbf{Y}_{M,1},\mathbf{Y}_{M,2})$ and $Cov(\mathbf{Y}_{M,1}^{2},\mathbf{Y}_{M,2}^{2})$
are $O(1/M)$. To do this, it is necessary to show that the denominator
of $\mathbf{Y}_{M,i}$ converges to a nonzero constant or to zero
sufficiently slowly. It is also necessary to show that the covariance
of the numerator is $O(1/M)$. Therefore, to bound $Cov(\mathbf{Y}_{M,1},\mathbf{Y}_{M,2})$,
we require bounds on the quantity $Cov\left[g\left(\lhatl l(\mathbf{X}_{i})\right),g\left(\lhatl{l'}(\mathbf{X}_{j})\right)\right]$
where $l,\, l'\in\bar{l}$. 

Define $\mathcal{M}(\mathbf{Z}):=\mathbf{Z}-\mathbb{E}\mathbf{Z}$,
$\ekl l(\mathbf{Z}):=\lhatl l(\mathbf{Z})-\mathbb{E}_{\mathbf{Z}}\left(\lhatl l(\mathbf{Z})\right)$,
and $\ehatl il(\mathbf{Z}):=\fhatl il(\mathbf{Z})-\mathbb{E}_{\mathbf{Z}}\fhatl il(\mathbf{Z})$.
Assuming $g$ is sufficiently smooth, a Taylor series expansion of
$g\left(\lhatl l(\mathbf{Z})\right)$ around $\mathbb{E}_{\mathbf{Z}}\lhatl l(\mathbf{Z})$
gives \[
g\left(\lhatl l(\mathbf{Z})\right)=\sum_{i=0}^{\lambda-1}\frac{g^{(i)}\left(\mathbb{E}_{\mathbf{Z}}\lhatl l(\mathbf{Z})\right)}{i!}\ekl l^{i}(\mathbf{Z})+\frac{g^{(\lambda)}\left(\mathbf{\xi_{Z}}\right)}{\lambda!}\ekl l^{\lambda}(\mathbf{Z}),\]
where $\mathbf{\xi_{Z}}\in\left(\mathbb{E}_{\mathbf{Z}}\ekl l(\mathbf{Z}),\ekl l(\mathbf{Z})\right)$.
We use this expansion to bound the covariance. The expected value
of the terms containing the derivatives of $g$ is controlled by assuming
that the densities are lower bounded. By assuming the densities are
sufficiently smooth, an expression for $\ekl l^{q}\left(\mathbf{Z}\right)$
in terms of powers and products of the density error terms $\ehatl 1l$
and $\ehatl 2l$ is obtained by expanding $\lhatl l(\mathbf{Z})$
around $\mathbb{E}_{\mathbf{Z}}\fhatl 1l(\mathbf{Z})$ and $\mathbb{E}_{\mathbf{Z}}\fhatl 2l(\mathbf{Z})$
and applying the binomial theorem. The expected value of products
of these density error terms is bounded by applying concentration
inequalities and conditional independence. Then the covariance between
$\ekl l^{q}(\mathbf{Z})$ terms is bounded by bounding the covariance
between powers and products of the density error terms by applying
Cauchy-Schwarz and other concentration inequalities. This gives the
following lemma which is proved in the appendices.
\begin{lem}
\label{lem:covariance_ek}Let $l,l'\in\bar{l}$ be fixed, $M_{1}=M_{2}=M$,
and $k(l)=l\sqrt{M}$. Let $\gamma_{1}(x),$ $\gamma_{2}(x)$ be arbitrary
functions with $1$ partial derivative wrt $x$ and $\sup_{x}|\gamma_{i}(x)|<\infty,\, i=1,\,2$
and let $1_{\{\cdot\}}$ be the indicator function. Let $\mathbf{X}_{i}$
and $\mathbf{X}_{j}$ be realizations of the density $f_{2}$ independent
of $\fhatl 1l$, $\fhatl 1{l'}$, $\fhatl 2l$, and $\fhatl 2{l'}$
and independent of each other when $i\neq j$. Then \[
Cov\left[\gamma_{1}(\mathbf{X}_{i})\ekl l^{q}(\mathbf{X}_{i}),\gamma_{2}(\mathbf{X}_{j})\ekl{l'}^{r}(\mathbf{X}_{j})\right]=\begin{cases}
o(1), & i=j\\
1_{\{q,r=1\}}c_{8}\left(\gamma_{1}(x),\gamma_{2}(x)\right)\left(\frac{1}{M}\right)+o\left(\frac{1}{M}\right), & i\neq j.\end{cases}\]

\end{lem}
Note that $k(l)$ is required to grow with $\sqrt{M}$ for Lemma~\ref{lem:covariance_ek}
to hold. Define $h_{l,g}(\mathbf{X})=g\left(\bE_{\mathbf{X}}\lhatl l(\mathbf{X})\right)$.
Lemma~\ref{lem:covariance_ek} can then be used to show that \[
Cov\left[g\left(\lhatl l(\mathbf{X}_{i})\right),g\left(\lhatl{l'}(\mathbf{X}_{j})\right)\right]=\begin{cases}
\bE\left[\mathcal{M}\left(h_{l,g}(\mathbf{X}_{i})\right)\mathcal{M}\left(h_{l',g}(\mathbf{X}_{i})\right)\right]+o(1), & i=j\\
c_{8}\left(h_{l,g'}(x),h_{l',g'}(x)\right)\left(\frac{1}{M}\right)+o\left(\frac{1}{M}\right), & i\neq j.\end{cases}\]

For the covariance of $\mathbf{Y}_{M,i}^{2}$ and $\mathbf{Y}_{M,j}^{2}$,
assume WLOG that $i=1$ and $j=2$. Then for $l,\, l',\, j,\, j'$
we need to bound the term \begin{equation}
Cov\left[\mathcal{M}\left(g\left(\lhatl l(\mathbf{X}_{1})\right)\right)\mathcal{M}\left(g\left(\lhatl{l'}(\mathbf{X}_{1})\right)\right),\mathcal{M}\left(g\left(\lhatl j(\mathbf{X}_{2})\right)\right)\mathcal{M}\left(g\left(\lhatl{j'}(\mathbf{X}_{2})\right)\right)\right].\label{eq:cov_squared}\end{equation}
For the case where $l=l'$ and $j=j'$, we can simply apply the previous
results to the functional $d(x)=\left(\mathcal{M}\left(g(x)\right)\right)^{2}$.
For the more general case, we need to show that \begin{equation}
Cov\left[\gamma_{1}(\mathbf{X}_{1})\ekl l^{s}(\mathbf{X}_{1})\ekl{l'}^{q}(\mathbf{X}_{1}),\gamma_{2}(\mathbf{X}_{2})\ekl j^{t}(\mathbf{X}_{2})\ekl{j'}^{r}(\mathbf{X}_{2})\right]=O\left(\frac{1}{M}\right).\label{eq:ek_squares}\end{equation}
To do this, bounds are required on the covariance of up to eight distinct
density error terms. Previous results can be applied by using Cauchy-Schwarz
when the sum of the exponents of the density error terms is greater
than or equal to 4. When the sum is equal to 3, we use the fact that
$k(l)=O(k(l'))$ combined with Markov's inequality to obtain a bound
of $O\left(1/M\right)$. Applying Eq.~\ref{eq:ek_squares} to the
term in Eq.~\ref{eq:cov_squared} gives the required bound to apply
Lemma~\ref{lem:clt_covariance}.

\subsection{Broad Implications of Theorem~\ref{thm:clt}}

To the best of our knowledge, Theorem~\ref{thm:clt} provides the
first results on the asymptotic distribution of an $f$-divergence
estimator with MSE convergence rate of $O\left(1/T\right)$ under
the setting of a finite number of samples from two unknown, non-parametric
distributions. This enables us to perform inference tasks on the class
of $f$-divergences (defined with smooth functions $g$) on smooth,
strictly lower bounded densities with finite support. Such tasks include
hypothesis testing and constructing a confidence interval on the error
exponents of the Bayes probability of error for a classification problem.
This greatly increases the utility of these divergence estimators.

Although we focused on a specific divergence estimator, we suspect
that our approach of showing that the components of the estimator
and their squares are asymptotically uncorrelated can be adapted to
derive central limit theorems for other divergence estimators that
satisfy similar assumptions (smooth $g$, and smooth, strictly lower
bounded densities with finite support). We speculate that this would
be easiest for estimators that are also based on $k$-nearest neighbors
such as in~\cite{poczos2011divergence} and~\cite{wang2009divergence}.
It is also possible that the approach can be adapted to other plug-in
estimator approaches such as in~\cite{singh2014exp} and~\cite{krishnamurthy2014div}.
However, the qualitatively different convex optimization approach
of divergence estimation in~\cite{nguyen2010divergence} may require
different methods.

\section{Experiments\label{sec:Experiments}}

We first apply the weighted ensemble estimator of divergence to simulated
data to verify the central limit theorem. We then use the estimator
to obtain confidence intervals on the error exponents of the Bayes
probability of error for the Iris data set from the UCI machine learning
repository~\cite{uci2013,fisher1936use}.

\subsection{Simulation}

To verify the central limit theorem of the ensemble method, we estimated
the KL divergence between two truncated normal densities restricted
to the unit cube. The densities have means $\bar{\mu}_{1}=0.7*\bar{1}_{d}$,
$\bar{\mu}_{2}=0.3*\bar{1}_{d}$ and covariance matrices $\sigma_{i}I_{d}$
where $\sigma_{1}=0.1,$ $\sigma_{2}=0.3,$ $\bar{1}_{d}$ is a $d$-dimensional
vector of ones, and $I_{d}$ is a $d$-dimensional identity matrix.
We show the Q-Q plot of the normalized optimally weighted ensemble
estimator of the KL divergence with $d=6$ and $1000$ samples from
each density in Fig.~\ref{fig:qqplot}. The linear relationship between
the quantiles of the normalized estimator and the standard normal
distribution validates Theorem~\ref{thm:clt}. 

\begin{SCfigure}

\includegraphics[width=0.53\textwidth]{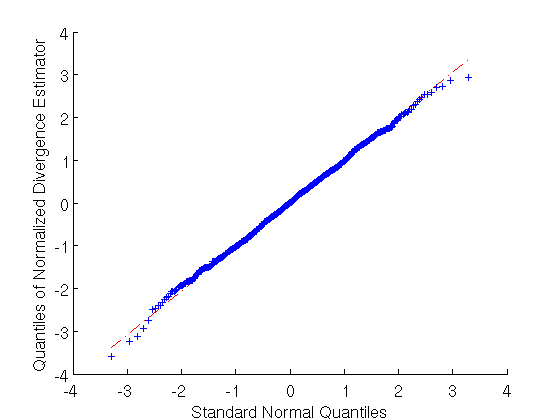}

\caption{Q-Q plot comparing quantiles from the normalized weighted ensemble estimator of the KL divergence (vertical axis) to the quantiles from the standard normal distribution (horizontal axis). The red line shows . The linearity of the Q-Q plot points validates the central limit theorem, Theorem.~\ref{thm:clt}, for the estimator. \label{fig:qqplot}}

\end{SCfigure}

\subsection{Probability of Error Estimation}

Our ensemble divergence estimator can be used to estimate a bound
on the Bayes probability of error~\cite{cover2006infotheory}. Suppose
we have two classes $C_{1}$ or $C_{2}$ and a random observation
$x$. Let the \emph{a priori }class probabilities be $w_{1}=Pr(C_{1})>0$
and $w_{2}=Pr(C_{2})=1-w_{1}>0$. Then $f_{1}$ and $f_{2}$ are the
densities corresponding to the classes $C_{1}$ and $C_{2}$, respectively.
The Bayes decision rule classifies $x$ as $C_{1}$ if and only if
$w_{1}f_{1}(x)>w_{2}f_{2}(x)$. The Bayes error $P_{e}^{*}$ is the
minimum average probability of error and is equivalent to \begin{eqnarray}
P_{e}^{*} & = & \int\min\left(Pr(C_{1}|x),Pr(C_{2}|x)\right)p(x)dx\nonumber \\
 & = & \int\min\left(w_{1}f_{1}(x),w_{2}f_{2}(x)\right)dx,\label{eq:PError}\end{eqnarray}
where $p(x)=w_{1}f_{1}(x)+w_{2}f_{2}(x)$. For $a,\, b>0$, we have
\[
\min(a,b)\leq a^{\alpha}b^{1-\alpha},\,\forall\alpha\in(0,1).\]
Replacing the minimum function in Eq.~\ref{eq:PError} with this
bound gives \begin{equation}
P_{e}^{*}\leq w_{1}^{\alpha}w_{2}^{1-\alpha}c_{\alpha}(f_{1}||f_{2}),\label{eq:chernoff1}\end{equation}
where $c_{\alpha}(f_{1}||f_{2})=\int f_{1}^{\alpha}(x)f_{2}^{1-\alpha}(x)dx$
is the Chernoff $\alpha$-coefficient. The Chernoff coefficient is
found by choosing the value of $\alpha$ that minimizes the right
hand side of Eq.~\ref{eq:chernoff1}: \[
c^{*}(f_{1}||f_{2})=c_{\alpha^{*}}(f_{1}||f_{2})=\min_{\alpha\in(0,1)}\int f_{1}^{\alpha}(x)f_{2}^{1-\alpha}(x)dx.\]
 Thus if $\alpha^{*}=\arg\min_{\alpha\in(0,1)}c_{\alpha}(f_{1}||f_{2})$,
an upper bound on the Bayes error is \begin{equation}
P_{e}^{*}\leq w_{1}^{\alpha^{*}}w_{2}^{1-\alpha^{*}}c^{*}(f_{1}||f_{2}).\label{eq:chernoff3}\end{equation}
Equation~\ref{eq:chernoff3} includes the form in Eq.~\ref{eq:fdivergences}
($g(x)=x^{\alpha}$). Thus we can use the optimally weighted ensemble
estimator described in Sec.~\ref{sec:Weighted_ensemble} to estimate
a bound on the Bayes error. In practice, we estimate $c_{\alpha}(f_{1}||f_{2})$
for multiple values of $\alpha$ (e.g. $0.01,0.02,\dots,0.99$) and
choose the minimum.

We estimated a bound on the pairwise Bayes error between the three
classes (Setosa, Versicolor, and Virginica) in the Iris data set~\cite{uci2013,fisher1936use}
and used bootstrapping to calculate confidence intervals. We compared
the bounds to the performance of a quadratic discriminant analysis
classifier (QDA) with 5-fold cross validation. The pairwise estimated
95\% confidence intervals and the misclassification rates of the QDA
are given in Table~\ref{tab:iris}. Note that the right endpoint
of the confidence interval is less than $1/50$ when comparing the
Setosa class to either of the other two classes. This is consistent
with the performance of the QDA and the fact that the Setosa class
is linearly separable from the other two classes. In contrast, the
right endpoint of the confidence interval is higher when comparing
the Versicolor and Virginica classes which are not linearly separable.
This is also consistent with the QDA performance. Thus the estimated
bounds provide a measure of the relative difficulty of distinguishing
between the classes, even though the small number of samples for each
class ($50$) limits the accuracy of the estimated bounds.

\begin{table}
\centering

\begin{tabular}{|c|c|c|c|}
\hline 
 & Setosa-Versicolor & Setosa-Virginica & Versicolor-Virginica\tabularnewline
\hline
\hline 
Estimated Confidence Interval & $(0,0.0013)$ & $(0,0.0002)$ & $(0,0.0726)$\tabularnewline
\hline 
QDA Misclassification Rate & $0$ & $0$ & $0.04$\tabularnewline
\hline
\end{tabular}

\caption{Estimated 95\% confidence intervals for the bound on the pairwise
Bayes error and the misclassification rate of a QDA classifier with
5-fold cross validation applied to the Iris dataset. The right endpoint
of the confidence intervals is nearly zero when comparing the Setosa
class to the other two classes while the right endpoint is much higher
when comparing the Versicolor and Virginica classes. This is consistent
with the QDA performance and the fact that the Setosa class is linearly
separable from the other two classes.\label{tab:iris}}

\end{table}

\section{Conclusion\label{sec:Conclusion}}

In this paper, we established the asymptotic normality for a weighted
ensemble estimator of $f$-divergence using $d$-dimensional truncated
$k$-nn density estimators. To the best of our knowledge, this gives
the first results on the asymptotic distribution of an $f$-divergence
estimator with MSE convergence rate of $O\left(1/T\right)$ under
the setting of a finite number of samples from two unknown, non-parametric
distributions. Future work includes simplifying the constants in front
of the convergence rates given in~\cite{moon2014isit} for certain
families of distributions, deriving Berry-Esseen bounds on the rate
of distributional convergence, extending the central limit theorem
to other divergence estimators, and deriving the nonasymptotic distribution
of the estimator.

\subsubsection*{Acknowledgments}

This work was partially supported by NSF grant CCF-1217880 and a NSF
Graduate Research Fellowship to the first author under Grant No. F031543.

\appendix

\section{Assumptions}

We use the same assumptions on the densities and the functional as
in~\cite{moon2014isit} and~\cite{arxiv2014div}. They are
\begin{itemize}
\item $(\mathcal{A}.0)$: Assume that $k_{i}=k_{0}M_{i}^{\beta}$ with $0<\beta<1$,
that $M_{2}=\alpha_{frac}T$ with $0<\alpha_{frac}<1$. 
\item $(\mathcal{A}.1)$: Assume there exist constants $\epsilon_{0},\epsilon_{\infty}$
such that $0<\epsilon_{0}\leq f_{i}(x)\leq\epsilon_{\infty}<\infty,\,\forall x\in S.$ 
\item $(\mathcal{A}.2)$: Assume that the densities $f_{i}$ have continuous
partial derivatives of order $d$ in the interior of $\mathcal{S}$
that are upper bounded. 
\item $(\mathcal{A}.3)$: Assume that $g$ has derivatives $g^{(j)}$ of
order $j=1,\dots,\max\{\lambda,d\}$ where $\lambda\beta>1$. 
\item $(\mathcal{A}.4$): Assume that $\left|g^{(j)}\left(f_{1}(x)/f_{2}(x)\right)\right|$,
$j=0,\ldots,\max\{\lambda,d\}$ are strictly upper bounded for $\epsilon_{0}\leq f_{i}(x)\leq\epsilon_{\infty}.$ 
\item $(\mathcal{A}.5)$: Let $\epsilon\in(0,1)$, $\delta\in(2/3,1)$,
and $\mathcal{C}(k)=\exp\left(-3k^{(1-\delta)}\right).$ For fixed
$\epsilon,$ define $p_{l,i}=(1-\epsilon)\epsilon_{0}\frac{k_{i}-1}{M_{i}}$,
$p_{u,i}=(1+\epsilon)\epsilon_{\infty}\frac{k_{i}-1}{M_{i}}$, $q_{l,i}=\frac{k_{i}-1}{M_{i}\bar{c}D^{d}},$
and $q_{u,i}=(1+\epsilon)\epsilon_{\infty}$ where $D$ is the diameter
of the support $S.$ Let $\mathbf{P}_{i}$ be a beta distributed random
variable with parameters $k_{i}$ and $M_{i}-k_{i}+1.$ Define $p_{l}=\frac{p_{l,1}}{p_{u,2}}$
and $p_{u}=\frac{p_{u,1}}{p_{l,2}}$. Assume that for $U(L)=g(L),\, g^{(3)}(L),$
and $g^{(\lambda)}(L),$

\begin{itemize}
\item $(i)\,\mathbb{E}\left[\sup_{L\in(p_{l},p_{u})}\left|U\left(L\frac{\mathbf{P}_{2}}{\mathbf{P}_{1}}\right)\right|\right]=G_{1}<\infty$, 
\item $(ii)\,\sup_{L\in\left(\frac{q_{l,1}}{q_{u,2}},\frac{q_{u,1}}{q_{l,2}}\right)}\left|U\left(L\right)\right|\mathcal{C}\left(k_{1}\right)\mathcal{C}\left(k_{2}\right)=G_{2}<\infty,$ 
\item $(iii)\,\mathbb{E}\left[\sup_{L\in\left(\frac{q_{l,1}}{p_{u,2}},\frac{q_{u,1}}{p_{l,2}}\right)}\left|U\left(L\mathbf{P}_{2}\right)\right|\mathcal{C}\left(k_{1}\right)\right]=G_{3}<\infty,$ 
\item $(iv)\,\mathbb{E}\left[\sup_{L\in\left(\frac{p_{l,1}}{q_{u,2}},\frac{p_{u,1}}{q_{l,2}}\right)}\left|U\left(\frac{L}{\mathbf{P}_{1}}\right)\right|\mathcal{C}\left(k_{2}\right)\right]=G_{4}<\infty,\,\forall M_{i}.$ 
\end{itemize}
\end{itemize}
Densities for which assumptions $(\mathcal{A}.0)-(\mathcal{A}.5)$
hold include the truncated Gaussian distribution and the Beta distribution
on the unit cube. Functions for which the assumptions hold include
$g(L)=-\ln L$ and $g(L)=L^{\alpha}.$

\section{Proof of Theorem 2}

We use Lemma~\ref{lem:clt_covariance} which is proved in~\cite{sricharan2012estimation}
and restate it here:
\begin{lem}
\label{lem:clt_covariance-1}Let the random variables $\{\mathbf{Y}_{M,i}\}_{i=1}^{N}$
belong to a zero mean, unit variance, interchangeable process for
all values of $M$. Assume that $Cov(\mathbf{Y}_{M,1},\mathbf{Y}_{M,2})$
and $Cov(\mathbf{Y}_{M,1}^{2},\mathbf{Y}_{M,2}^{2})$ are $O(1/M)$.
Then the random variable \begin{equation}
\mathbf{S}_{N,M}=\frac{\sum_{i=1}^{N}\mathbf{Y}_{M,i}}{\sqrt{\var\left[\sum_{i=1}^{N}\mathbf{Y}_{M,i}\right]}}\label{eq:sum-1}\end{equation}
converges in distribution to a standard normal random variable.
\end{lem}
For simplicity, let $M_{1}=M_{2}=M$ and $\lhatl l:=\hat{\mathbf{L}}_{k(l),k(l)}$.
Define \[
\mathbf{Y}_{M,i}=\frac{\sum_{l\in\bar{l}}w(l)g\left(\lhatl l(\mathbf{X}_{i})\right)-\bE\left[\sum_{l\in\bar{l}}w(l)g\left(\lhatl l(\mathbf{X}_{i})\right)\right]}{\sqrt{\var\left[\sum_{l\in\bar{l}}w(l)g\left(\lhatl l(\mathbf{X}_{i})\right)\right]}}.\]
 Then from Eq.~\ref{eq:sum-1}, we have that \[
\mathbf{S}_{N,M}=\frac{\hat{\mathbf{G}}_{w}-\bE\left[\hat{\mathbf{G}}_{w}\right]}{\sqrt{\var\left[\hat{\mathbf{G}}_{w}\right]}}.\]
 Thus it is sufficient to show from Lemma~\ref{lem:clt_covariance-1}
that $Cov(\mathbf{Y}_{M,1},\mathbf{Y}_{M,2})$ and $Cov(\mathbf{Y}_{M,1}^{2},\mathbf{Y}_{M,2}^{2})$
are $O(1/M)$. To do this, it is necessary to show that the denominator
of $\mathbf{Y}_{M,i}$ converges to a nonzero constant or to zero
sufficiently slowly. Note that the numerator and denominator of $\mathbf{Y}_{M,i}$
are, respectively,\[
\sum_{l\in\bar{l}}w(l)g\left(\lhatl l(\mathbf{X}_{i})\right)-\bE\left[\sum_{l\in\bar{l}}w(l)g\left(\lhatl l(\mathbf{X}_{i})\right)\right]\]
\begin{equation}
=\sum_{l\in\bar{l}}w(l)\left(g\left(\lhatl l(\mathbf{X}_{i})\right)-\bE\left[g\left(\lhatl l(\mathbf{X}_{i})\right)\right]\right),\label{eq:y_num}\end{equation}
\[
\sqrt{\var\left[\sum_{l\in\bar{l}}w(l)g\left(\lhatl l(\mathbf{X}_{i})\right)\right]}\]
\begin{equation}
=\sqrt{\sum_{l\in\bar{l}}\sum_{l'\in\bar{l}}w(l)w(l')Cov\left(g\left(\lhatl l(\mathbf{X}_{i})\right),g\left(\lhatl{l'}(\mathbf{X}_{i})\right)\right)}.\label{eq:y_den}\end{equation}
 Therefore, to bound $Cov(\mathbf{Y}_{M,1},\mathbf{Y}_{M,2})$, we
require bounds on the quantity $Cov\left[g\left(\lhatl l(\mathbf{X}_{i})\right),g\left(\lhatl{l'}(\mathbf{X}_{j})\right)\right]$. 

Some preliminary work is required before we can directly tackle this
quantity. Define $\mathcal{M}(\mathbf{Z}):=\mathbf{Z}-\mathbb{E}\mathbf{Z}$,
$\ekl l(\mathbf{Z}):=\lhatl l(\mathbf{Z})-\mathbb{E}_{\mathbf{Z}}\left(\lhatl l(\mathbf{Z})\right)$,
and $\ehatl il(\mathbf{Z}):=\fhatl il(\mathbf{Z})-\mathbb{E}_{\mathbf{Z}}\fhatl il(\mathbf{Z})$.
By forming a Taylor series expansion of $g\left(\lhatl l(\mathbf{Z})\right)$
around $\mathbb{E}_{\mathbf{Z}}\lhatl l(\mathbf{Z})$, we get \[
g\left(\lhatl l(\mathbf{Z})\right)=\sum_{i=0}^{\lambda-1}\frac{g^{(i)}\left(\mathbb{E}_{\mathbf{Z}}\lhatl l(\mathbf{Z})\right)}{i!}\ekl l^{i}(\mathbf{Z})+\frac{g^{(\lambda)}\left(\mathbf{\xi_{Z}}\right)}{\lambda!}\ekl l^{\lambda}(\mathbf{Z}),\]
 where $\mathbf{\xi_{Z}}\in\left(\mathbb{E}_{\mathbf{Z}}\ekl l(\mathbf{Z}),\ekl l(\mathbf{Z})\right)$.
Let $\Psi(\mathbf{Z})=g^{(\lambda)}\left(\mathbf{\xi_{Z}}\right)/\lambda!$
and \begin{eqnarray*}
\gtay pil & := & \mathcal{M}\left(g\left(\bE_{\mathbf{X}_{i}}\lhatl l(\mathbf{X}_{i})\right)\right),\\
\mathbf{q}_{i}^{(l)} & := & \mathcal{M}\left(g'\left(\mathbb{E}_{\mathbf{X}_{i}}\lhatl l\left(\mathbf{X}_{i}\right)\right)\ekl l\left(\mathbf{X}_{i}\right)\right),\\
\gtay ril & := & \mathcal{M}\left(\sum_{j=2}^{\lambda-1}\frac{g^{(j)}\left(\mathbb{E}_{\mathbf{X}_{i}}\lhatl l\left(\mathbf{X}_{i}\right)\right)}{j!}\ekl l^{j}\left(\mathbf{X}_{i}\right)\right),\\
\gtay sil & := & \mathcal{M}\left(\Psi\left(\mathbf{X}_{i}\right)\ekl l^{\lambda}\left(\mathbf{X}_{i}\right)\right).\end{eqnarray*}
Then\[
Cov\left[g\left(\lhatl l(\mathbf{X}_{i})\right),g\left(\lhatl{l'}(\mathbf{X}_{j})\right)\right]\]
 \begin{equation}
=\bE\left[\left(\gtay pil+\gtay qil+\gtay ril+\gtay sil\right)\left(\gtay pj{l'}+\gtay qj{l'}+\gtay rj{l'}+\gtay sj{l'}\right)\right].\label{eq:cov_g}\end{equation}

To obtain expressions for $\ekl l^{i}\left(\mathbf{Z}\right)$, we
expand $\lhatl l(\mathbf{Z})$ around $\mathbb{E}_{\mathbf{Z}}\fhatl 1l(\mathbf{Z})$
and $\mathbb{E}_{\mathbf{Z}}\fhatl 2l(\mathbf{Z})$:

\begin{eqnarray}
\frac{\fhatl 1l(\mathbf{Z})}{\fhatl 2l(\mathbf{Z})} & = & \frac{\mathbb{E}_{\mathbf{Z}}\fhatl 1l(\mathbf{Z})}{\mathbb{E}_{\mathbf{Z}}\fhatl 2l(\mathbf{Z})}+\frac{\ehatl 1l(\mathbf{Z})}{\mathbb{E}_{\mathbf{Z}}\fhatl 2l(\mathbf{Z})}-\mathbb{E}_{\mathbf{Z}}\fhatl 1l(\mathbf{Z})\frac{\ehatl 2l(\mathbf{Z})}{\left(\mathbb{E}_{\mathbf{Z}}\fhatl 2l(\mathbf{Z})\right)^{2}}\nonumber \\
 &  & -\frac{\ehatl 1l(\mathbf{Z})\ehatl 2l(\mathbf{Z})}{\left(\mathbb{E}_{\mathbf{Z}}\fhatl 2l(\mathbf{Z})\right)^{2}}+\mathbb{E}_{\mathbf{Z}}\fhatl 1l(\mathbf{Z})\frac{\ehatl 2l^{2}(\mathbf{Z})}{2\left(\mathbb{E}_{\mathbf{Z}}\fhatl 2l(\mathbf{Z})\right)^{3}}\nonumber \\
 &  & +\frac{\ehatl 1l(\mathbf{Z})\ehatl 2l^{2}(\mathbf{Z})}{2\left(\mathbb{E}_{\mathbf{Z}}\fhatl 2l(\mathbf{Z})\right)^{3}}+o\left(\ehatl 2l^{2}(\mathbf{Z})+\ehatl 1l(\mathbf{Z})\ehatl 2l^{2}(\mathbf{Z})\right)\label{eq:lhat_taylor}\\
 & = & \frac{\mathbb{E}_{\mathbf{Z}}\fhatl 1l(\mathbf{Z})}{\mathbb{E}_{\mathbf{Z}}\fhatl 2l(\mathbf{Z})}+h(\ehatl 1l(\mathbf{Z}),\ehatl 2l(\mathbf{Z})).\nonumber \end{eqnarray}
Let $\mathbf{h}(\mathbf{Z})=h(\ehatl 1l(\mathbf{Z}),\ehatl 2l(\mathbf{Z})).$
Thus $\ekl l(\mathbf{Z})=\frac{\mathbb{E}_{\mathbf{Z}}\fhatl 1l(\mathbf{Z})}{\mathbb{E}_{\mathbf{Z}}\fhatl 2l(\mathbf{Z})}-\mathbb{E}_{\mathbf{Z}}\lhatl l(\mathbf{Z})+\mathbf{h}(\mathbf{Z}).$
By the binomial theorem, \begin{equation}
\ekl l^{q}(\mathbf{Z})=\sum_{j=0}^{q}a_{q,j}\left(\frac{\mathbb{E}_{\mathbf{Z}}\fhatl 1l(\mathbf{Z})}{\mathbb{E}_{\mathbf{Z}}\fhatl 2l(\mathbf{Z})}-\mathbb{E}_{\mathbf{Z}}\lhatl l(\mathbf{Z})\right)^{q-j}\mathbf{h}^{j}(\mathbf{Z}),\label{eq:f_binomial}\end{equation}
where $a_{q,j}$ is the binomial coefficient. Using a Taylor series
expansion of $\frac{1}{x}$ about $\mathbb{E}_{\mathbf{Z}}\fhat 2(\mathbf{Z})$,
\begin{eqnarray}
\mathbb{E}_{\mathbf{Z}}\frac{1}{\fhat 2(\mathbf{Z})} & = & \mathbb{E}_{\mathbf{Z}}\left[\frac{1}{\mathbb{E}_{\mathbf{Z}}\fhat 2(\mathbf{Z})}-\frac{\ehat 2}{\left(\mathbb{E}_{\mathbf{Z}}\fhat 2(\mathbf{Z})\right)^{2}}+\frac{\ehat 2^{2}}{2\xi_{2,\mathbf{Z}}}\right]\nonumber \\
 & = & \frac{1}{\mathbb{E}_{\mathbf{Z}}\fhat 2(\mathbf{Z})}+\frac{\left(\var_{\mathbf{Z}}\left[\fhat 2(\mathbf{Z})\right]\right)}{2\xi_{2,\mathbf{Z}}}\nonumber \\
 & = & \frac{1}{\mathbb{E}_{\mathbf{Z}}\fhat 2(\mathbf{Z})}+c_{3,2}(\mathbf{Z})\left(\frac{1}{k_{2}}\right),\label{eq:kde}\end{eqnarray}
where $\xi_{2,\mathbf{Z}}\in\left(\mathbb{E}_{\mathbf{Z}}\fhat 2(\mathbf{Z}),\fhat 2(\mathbf{Z})\right)$
from the mean value thoerem and we use the fact that the variance
of the kernel density estimate converges to zero with rate $\frac{1}{M\sigma}$
where $\sigma=O\left(\frac{k(l)}{M}\right)$. Sricharan et al~\cite{sricharan2013ensemble}
showed that for a truncated uniform kernel density estimator with
bandwidth $\left(k/M\right)^{1/d}$, $\mathbb{E}_{\mathbf{Z}}\fhatl il(\mathbf{Z})=f_{i}(\mathbf{Z})+\sum_{j=1}^{d}c_{i,j,k(l)}(\mathbf{Z})\left(\frac{k(l)}{M}\right)^{j/d}+o\left(\frac{k(l)}{M}\right)=f_{i}(\mathbf{Z})+c_{1,i}(\mathbf{Z},k(l),M)=f_{i}(\mathbf{Z})+o(1).$
It can then be shown that the $k$-nn density estimator converges
to a truncated uniform kernel density estimator~\cite{kumar2012thesis}.\textbf{\emph{
}}Thus the result holds for the $k$-nn density estimator as well.
Combining this with Eq.~\ref{eq:kde} gives \begin{eqnarray}
 &  & \left(\frac{\mathbb{E}_{\mathbf{Z}}\fhatl 1l(\mathbf{Z})}{\mathbb{E}_{\mathbf{Z}}\fhatl 2l(\mathbf{Z})}-\mathbb{E}_{\mathbf{Z}}\lhat(\mathbf{Z})\right)^{q}\nonumber \\
 & = & \left(\mathbb{E}_{\mathbf{Z}}\fhatl 1l(\mathbf{Z})c_{3,2}(\mathbf{Z})\left(\frac{1}{k(l)}\right)\right)^{q}\\
 & = & \left(f_{1}(\mathbf{Z})c_{3,2}(\mathbf{Z})\left(\frac{1}{k(l)}\right)+\sum_{j=1}^{d}c_{1,j,k(l)}\left(\frac{k(l)}{M}\right)^{\frac{j}{d}}\left(\frac{1}{k(l)}\right)+o\left(\frac{1}{M}\right)\right)^{q}\nonumber \\
 & = & 1_{\{q=1\}}c_{3}(\mathbf{Z})\left(\frac{1}{k(l)}\right)+1_{\{q\geq2\}}O\left(\frac{1}{k(l)^{q}}\right)+o\left(\frac{1}{M}\right)=:\bz q(\mathbf{Z}).\label{eq:bias_final1}\end{eqnarray}
Combining Eqs.~\ref{eq:lhat_taylor}, \ref{eq:f_binomial}, and \ref{eq:bias_final1}
gives \begin{eqnarray}
\ekl l^{q}(\mathbf{Z}) & = & \bz ql(\mathbf{Z})+\bz{q-1}l^{1_{\{q\geq2\}}}(\mathbf{Z})a_{q,1}\mathbf{h}(\mathbf{Z})+1_{\{q\geq2\}}\bz{q-2}l^{^{1_{\{q\geq3\}}}}(\mathbf{Z})a_{q,2}\mathbf{h}^{2}(\mathbf{Z})\nonumber \\
 &  & +1_{\{q\geq3\}}\bz{q-3}l^{1_{\{q\geq4\}}}\mathbf{Z})O\left(\mathbf{h}^{3}(\mathbf{Z})\right)\label{eq:Fk_q}\end{eqnarray}
where \begin{eqnarray*}
\mathbf{h}(\mathbf{Z}) & = & \frac{\ehatl 1l(\mathbf{Z})}{\ez\fhatl 2l(\mathbf{Z})}-\frac{\ez\fhatl 1l(\mathbf{Z})}{\left(\ez\fhatl 2l(\mathbf{Z})\right)^{2}}\ehatl 2l(\mathbf{Z})-\frac{\ehatl 1l(\mathbf{Z})\ehatl 2l(\mathbf{Z})}{\left(\ez\fhatl 2l(\mathbf{Z})\right)^{2}}\\
 &  & +\frac{\ez\fhatl 1l(\mathbf{Z})}{2\left(\ez\fhatl 2l(\mathbf{Z})\right)^{3}}\ehatl 2l^{2}(\mathbf{Z})+o\left(\ehatl 2l^{2}(\mathbf{Z})\right),\\
\mathbf{h}^{2}(\mathbf{Z}) & = & \frac{\ehatl 1l^{2}(\mathbf{Z})}{\left(\ez\fhatl 2l(\mathbf{Z})\right)^{2}}+\frac{\left(\ez\fhatl 1l(\mathbf{Z})\right)^{2}}{\left(\ez\fhatl 2l(\mathbf{Z})\right)^{4}}\ehatl 2l^{2}(\mathbf{Z})\\
 &  & +O\left(\ehatl 1l(\mathbf{Z})\ehat 2(\mathbf{Z})+\ehatl 2l^{3}(\mathbf{Z})\right),\\
O\left(\mathbf{h}^{3}(\mathbf{Z})\right) & = & O\left(\ehatl 1l^{3}(\mathbf{Z})+\ehatl 2l^{3}(\mathbf{Z})+\ehatl 1l^{2}(\mathbf{Z})\ehatl 2l^{2}(\mathbf{Z})\right).\end{eqnarray*}

We now obtain bounds on the expected value of products of the $\ehatl il$
terms:
\begin{lem}
\label{lem:ekhat}Let $l,l'\in\bar{l}$ be fixed, $M_{1}=M_{2}=M$,
and $k(l)=l\sqrt{M}$. Let $\gamma(z)$ be an arbitrary function with
$\sup_{z}|\gamma(z)|<\infty.$ Let $\mathbf{Z}$ be a realization
of the density $f_{2}$ independent of $\fhatl il$ and $\fhatl i{l'}$
for $i=1,2$. Then,\begin{eqnarray}
\mathbb{E}\left[\gamma(\mathbf{Z})\ehatl il^{q}(\mathbf{Z})\right] & = & \begin{cases}
1_{\{q=2\}}\left(c_{2,i}(\gamma(z))\left(\frac{1}{k(l)}\right)+o\left(\frac{1}{k(l)}\right)\right)+1_{\{q\geq3\}}O\left(\frac{1}{k(l)^{\frac{q}{2}}}\right), & q\geq2\\
0, & q=1,\end{cases}\label{eq:moment}\end{eqnarray}
\begin{equation}
\mathbb{E}\left[\gamma(\mathbf{Z})\ehatl il^{q}(\mathbf{Z})\ehatl i{l'}^{r}(\mathbf{Z})\right]=\begin{cases}
O\left(\frac{1}{k(l)^{\frac{q}{2}}k(l')^{\frac{r}{2}}}\right), & q+r\geq2\\
0, & otherwise\end{cases}\label{eq:cross_moment}\end{equation}
\begin{eqnarray}
 &  & \mathbb{E}\left[\gamma(\mathbf{Z})\ehatl 1l^{q}(\mathbf{Z})\ehatl 1{l'}^{q'}(\mathbf{Z})\ehatl 2l^{r}(\mathbf{Z})\ehatl 2{l'}^{r'}(\mathbf{Z})\right]\nonumber \\
 & = & \begin{cases}
0, & q+q'=1\, or\, r+r'=1\\
O\left(\frac{1}{k(l)^{\frac{q+r}{2}}k(l')^{\frac{q'+r'}{2}}}\right), & otherwise\end{cases}\label{eq:moment3}\\
\mathbb{E}\left[\gamma(\mathbf{Z})\ekl l^{q}(\mathbf{Z})\right] & = & 1_{\{q=1\}}O\left(\frac{1}{k(l)}\right)+1_{\{q\geq2\}}O\left(\frac{1}{k(l)^{\frac{q}{2}}}\right).\label{eq:moment4}\end{eqnarray}
\end{lem}
\begin{proof}
For $i=2$, Eq.~\ref{eq:moment} is given and proved as Lemma 5 in~\cite{sricharan2013ensemble}
where the density estimator is a truncated uniform kernel density
estimator with bandwidth $\left(k(l)/M\right)^{1/d}$. The proof uses
concentration inequalities to bound $\ez\ehatl 2l^{q}(\mathbf{Z})$
in terms of $k(l)$. Then since the truncated uniform kernel density
estimator converges to the $k$-nn estimator, it holds for the $k$-nn
estimator as well. For $i=1,$ the proof follows the same procedure
but results in a different constant.

Equation~\ref{eq:cross_moment} is proved in a similar manner. Let
$S_{l}(X):=\left\{ Y\in\mathcal{S}:||X-Y||_{\infty}\leq\left(k(l)/M\right)^{1/d}/2\right\} $,
$V_{l}(X):=\int_{S_{l}(X)}dz$, $U_{i,l}(X):=Pr(\mathbf{Z}\in S_{l}(X))$
where $\mathbf{Z}$ is drawn from $f_{i}$, and $\mathbf{1}_{i,l}(X)$
denote the number of samples from the $i$th distribution that fall
in $S_{l}(X)$; i.e. the number of samples from$\left\{ \mathbf{Y}_{1},\dots,\mathbf{Y}_{M}\right\} $
if $i=1$ or $\left\{ \mathbf{X}_{N+1},\dots,\mathbf{X}_{N+M}\right\} $
if $i=2$ that fall in $S_{l}(X)$. The uniform kernel density estimator
is then \[
\ft il(X)=\frac{\mathbf{1}_{i,l}(X)}{MV_{l}(X)}.\]
Let $\natural_{i,l}(X)$ denote the event $(1-p_{k(l)})MU_{i,l}(X)<\mathbf{1}_{i,l}(X)<(1+p_{k(l)})MU_{i,l}(X)$,
where $p_{k(l)}=1/k(l)^{\delta/2}$. It can be shown~\cite{sricharan2013ensemble}
using standard Chernoff inequalities that $Pr(\natural_{i,l}^{C}(X))=O\left(e^{-p_{k(l)}^{2}k(l)}\right)$
and that under the event $\natural_{i,l}(X)$, $\ehatl il=O(1/(k^{\delta/2}))$.
Thus \begin{eqnarray*}
\mathbb{E}\left[\gamma(\mathbf{Z})\ehatl il^{q}(\mathbf{Z})\ehatl i{l'}^{r}(\mathbf{Z})\right] & = & \bE\left[\gamma(\mathbf{Z})1_{\natural_{i,l}(X)\cap\natural_{i,l'}(X)}\ehatl il^{q}(\mathbf{Z})\ehatl i{l'}^{r}(\mathbf{Z})\right]\\
 &  & +\bE\left[\gamma(\mathbf{Z})1_{\{\natural_{i,l}(X)\cap\natural_{i,l'}(X)\}^{C}}\ehatl il^{q}(\mathbf{Z})\ehatl i{l'}^{r}(\mathbf{Z})\right]\\
 & = & O\left(\frac{1}{k(l)^{\frac{q}{2}}k(l')^{\frac{r}{2}}}\right),\end{eqnarray*}
where we use the fact that $\delta$ can be chosen arbitrarily close
to $1$.

For Eq.~\ref{eq:moment3}, note that due to conditional independence
and Eq.~\ref{eq:cross_moment}, \begin{eqnarray*}
 &  & \mathbb{E}\left[\gamma(\mathbf{Z})\ehatl 1l^{q}(\mathbf{Z})\ehatl 1{l'}^{q'}(\mathbf{Z})\ehatl 2l^{r}(\mathbf{Z})\ehatl 2{l'}^{r'}(\mathbf{Z})\right]\\
 & = & \mathbb{E}\left[\gamma(\mathbf{Z})\bE_{\mathbf{Z}}\left[\ehatl 1l^{q}(\mathbf{Z})\ehatl 1{l'}^{q'}(\mathbf{Z})\right]\bE_{\mathbf{Z}}\left[\ehatl 2l^{r}(\mathbf{Z})\ehatl 2{l'}^{r'}(\mathbf{Z})\right]\right]\\
 & = & \mathbb{E}\left[\gamma(\mathbf{Z})\left(O\left(\frac{1}{k(l)^{\frac{q+r}{2}}k(l')^{\frac{q'+r'}{2}}}\right)\right)\right]\\
 & = & O\left(\frac{1}{k(l)^{\frac{q+r}{2}}k(l')^{\frac{q'+r'}{2}}}\right).\end{eqnarray*}

Equation~\ref{eq:moment4} is obtained by applying Eqs.~\ref{eq:moment}
and \ref{eq:cross_moment} to Eq.~\ref{eq:Fk_q}.
\end{proof}
Lemma~\ref{lem:covariance_ek} provides bounds on the covariance
between the $\ekl l^{q}(\mathbf{Z})$ terms and we restate it here
along with its proof:
\begin{lem}
\label{lem:covariance_ek-1}Let $l,l'\in\bar{l}$ be fixed, $M_{1}=M_{2}=M$,
and $k(l)=l\sqrt{M}$. Let $\gamma_{1}(x),$ $\gamma_{2}(x)$ be arbitrary
functions with $1$ partial derivative wrt $x$ and $\sup_{x}|\gamma_{i}(x)|<\infty,\, i=1,\,2.$
Let $\mathbf{X}_{i}$ and $\mathbf{X}_{j}$ be realizations of the
density $f_{2}$ independent of $\fhatl 1l$, $\fhatl 1{l'}$, $\fhatl 2l$,
and $\fhatl 2{l'}$ and independent of each other when $i\neq j$.
Then \[
Cov\left[\gamma_{1}(\mathbf{X}_{i})\ekl l^{q}(\mathbf{X}_{i}),\gamma_{2}(\mathbf{X}_{j})\ekl{l'}^{r}(\mathbf{X}_{j})\right]=\begin{cases}
o(1), & i=j\\
1_{\{q=1,r=1\}}c_{8}\left(\gamma_{1}(x),\gamma_{2}(x)\right)\left(\frac{1}{M}\right)+o\left(\frac{1}{M}\right), & i\neq j.\end{cases}\]
\end{lem}
\begin{proof}
Throughout the following, assume that $\mathbf{X}$ and $\mathbf{Y}$
are realizations of the density $f_{2}$ independent of each other
and $\fhatl 1l$, $\fhatl 1{l'}$, $\fhatl 2l$, and $\fhatl 2{l'}$.
First consider the case where $i=j$. By Cauchy-Schwarz and Eq.~\ref{eq:moment},
\begin{equation}
Cov\left[\gamma_{1}(\mathbf{X})\ehatl il^{q}(\mathbf{X}),\gamma_{2}(\mathbf{X})\ehatl i{l'}^{r}(\mathbf{X})\right]=O\left(\frac{1}{M^{\frac{q+r}{4}}}\right).\label{eq:cov_ells_inside}\end{equation}
By Eq.~\ref{eq:cross_moment} and Eq.~\ref{eq:moment3}, \begin{equation}
Cov\left[\gamma_{1}(\mathbf{X})\ehatl 1l^{q}(\mathbf{X})\ehatl 2l^{r}(\mathbf{X}),\gamma_{2}(\mathbf{X})\ehatl 1{l'}^{q'}(\mathbf{X})\ehatl 2{l'}^{r'}(\mathbf{X})\right]=O\left(\frac{1}{M^{\frac{q+r+q'+r'}{4}}}\right).\label{eq:cov_ells_inside2}\end{equation}
Applying Eqs.~\ref{eq:cov_ells_inside} and~\ref{eq:cov_ells_inside2}
to Eq.~\ref{eq:Fk_q} completes the proof for this case.

We'll now prove the case where $i\neq j$. Define $\Psi(l,l')=\left\{ \left\Vert X-Y\right\Vert _{1}\geq2\left(\frac{\max(k(l),k(l'))}{M}\right)^{\frac{1}{d}}\right\} $.
For a fixed pair of points $\{X,Y\}\in\Psi(l,l')$, \begin{equation}
Cov\left[\ehatl il^{q}(X),\ehatl i{l'}^{r}(Y)\right]=1_{\{q=r=1\}}\left(\frac{-f_{i}(X)f_{i}(Y)}{M}\right)+o\left(\frac{1}{M}\right).\label{eq:cov_ells_split}\end{equation}
This can be shown in the same way as in the proof of Lemma 6 in~\cite{sricharan2013ensemble}
for a truncated uniform kernel density estimator. This is done by
recognizing that for $\{X,Y\}\in\Psi(l,l')$, the functions $\mathbf{1}_{i,l}(X)$
and $\mathbf{1}_{i,l'}(Y)$ are distributed jointly as a multinomial
random variable with parameters $M$, $U_{i,l}(X)$, $U_{i,l'}(Y)$
and $1-U_{i,l}(X)-U_{i,l'}(Y)$. Equation~\ref{eq:cov_ells_split}
is then established by using the concentration inequality for the
high probability event of $\natural_{i,l}(X)\cap\natural_{i,l'}(Y)$
and then relating the functions $\mathbf{1}_{i,l}(X)$ and $\mathbf{1}_{i,l'}(Y)$
to two binomial random variables with parameters $\{U_{i,l}(X),M-q\}$
and $\{U_{i,l'}(Y),M-r\}$, respectively. Note that the relationship
holds whether $l=l'$ or $l\neq l'$. For fixed $\{X,Y\}\in\Psi(l,l')^{C}$,
Cauchy-Schwarz and Eq.~\ref{eq:moment} give \begin{equation}
Cov\left[\ehatl il^{q}(X),\ehatl i{l'}^{r}(Y)\right]=O\left(\frac{1}{k(l)^{\frac{q}{2}}k(l')^{\frac{r}{2}}}\right).\label{eq:cov_ells_split2}\end{equation}
From Eqs.~\ref{eq:cov_ells_split} and~\ref{eq:cov_ells_split2},
we have that 

\begin{equation}
Cov\left[\gamma_{1}(\mathbf{X})\ehatl il^{q}(\mathbf{X}),\gamma_{2}(\mathbf{Y})\ehatl i{l'}^{r}(\mathbf{Y})\right]=1_{\{q=r=1\}}c_{7,i}(\gamma_{1}(x),\gamma_{2}(x))\left(\frac{1}{M}\right)+o\left(\frac{1}{M}\right).\label{eq:cov_ells_final}\end{equation}
This is proved in the same way as in the proof of Lemma 8 in~\cite{sricharan2013ensemble}
by splitting the covariances into the cases where $\{\mathbf{X},\mathbf{Y}\}\in\Psi(l,l')$
and $\{\mathbf{X},\mathbf{Y}\}\in\Psi(l,l')^{C}$. For the first case,
the bound falls clearly from Eq.~\ref{eq:cov_ells_split}. For the
second case, the bound holds with Eq.~\ref{eq:cov_ells_split2} since
$\int_{\Psi(l,l')^{C}}dy=2^{d}\frac{\max(k(l),k(l'))}{M}$. 

Now let $E_{0}=\{s,q,t,r\geq1\}$, $E_{1,1}=\{s=0,q\geq2,t\geq1,r\geq1\}\cup\{s\geq1,q\geq1,t=0,r\geq2\}$,
and $E_{1,2}=\{s\geq2,q=0,t\geq1,r\geq1\}\cup\{s\geq1,q\geq1,t\geq2,r=0\}$.
For fixed $X$, $Y$, we have by Eqs.~\ref{eq:moment} and~\ref{eq:cross_moment}
and conditional independence when $E_{0},$ $E_{1,1}$, or $E_{1,2}$
hold that \begin{align*}
Cov\left[\gamma_{1}(X)\ehatl 1l^{s}(X)\ehatl 2l^{q}(X),\gamma_{2}(Y)\ehatl 1{l'}^{t}(Y)\ehatl 2{l'}^{r}(Y)\right]\end{align*}

\begin{eqnarray}
 & = & \mathbb{E}\left[\gamma_{1}(X)\ehatl 1l^{s}(X)\ehatl 2l^{q}(X)\gamma_{2}(Y)\ehatl 1{l'}^{t}(Y)\ehatl 2{l'}^{r}(Y)\right]\nonumber \\
 &  & -\mathbb{E}\left[\gamma_{1}(X)\ehatl 1l^{s}(X)\ehatl 2l^{q}(X)\right]\mathbb{E}\left[\gamma_{2}(Y)\ehatl 1{l'}^{t}(Y)\ehatl 2{l'}^{r}(Y)\right]\nonumber \\
 & = & \gamma_{1}(X)\gamma_{2}(Y)\mathbb{E}\left[\ehatl 1l^{s}(X)\ehatl 1{l'}^{t}(Y)\right]\mathbb{E}\left[\ehatl 2l^{q}(X)\ehatl 2{l'}^{r}(Y)\right]\nonumber \\
 &  & +1_{\{q,r,s,t\neq1\}}O\left(\frac{1}{k(l)^{\frac{q+s}{2}}k(l')^{\frac{r+t}{2}}}\right).\label{eq:cov_inside1}\end{eqnarray}
Now $\mathbb{E}\left[\ehatl il^{s}(X)\ehatl i{l'}^{t}(Y)\right]=Cov\left[\ehatl il^{s}(X),\ehatl i{l'}^{t}(Y)\right]+\mathbb{E}\left[\ehatl il^{s}(X)\right]\mathbb{E}\left[\ehatl i{l'}^{t}(Y)\right]$.
By Eqs.~\ref{eq:moment}, \ref{eq:cov_ells_split}, and \ref{eq:cov_ells_split2}
this gives (when $s,t\geq1$) \begin{eqnarray}
\mathbb{E}\left[\ehatl il^{s}(X)\ehatl i{l'}^{t}(Y)\right] & = & 1_{\{s,t\geq2\}}O\left(\frac{1}{k(l)^{\frac{s}{2}}k(l')^{\frac{t}{2}}}\right)\nonumber \\
 &  & +\begin{cases}
1_{\left\{ s=t=1\right\} }\left(\frac{-f_{i}(X)f_{i}(Y)}{M}\right)+o\left(\frac{1}{M}\right), & \{X,Y\}\in\Psi(l,l')\\
O\left(\frac{1}{k(l)^{\frac{s}{2}}k(l')^{\frac{t}{2}}}\right), & \{X,Y\}\in\Psi(l,l')^{c}.\end{cases}\label{eq:cov_inside2}\end{eqnarray}
Now \[
\mathbb{E}\left[Cov_{\mathbf{X,Y}}\left[\gamma_{1}(\mathbf{X})\ehatl 1l^{s}(\mathbf{X})\ehatl 2l^{q}(\mathbf{X}),\gamma_{2}(\mathbf{Y})\ehatl 1{l'}^{t}(\mathbf{Y})\ehatl 2{l'}^{r}(\mathbf{Y})\right]\right]=I_{1}+I_{2},\]
where \[
I_{1}=\mathbb{E}\left[1_{\{\mathbf{X},\mathbf{Y}\}\in\Psi(l,l')^{C}}\gamma_{1}(\mathbf{X})\gamma_{2}(\mathbf{Y})Cov_{\mathbf{X,Y}}\left[\ehatl 1l{}^{s}(\mathbf{X})\ehatl 2l^{q}(\mathbf{X}),\ehatl 1{l'}^{t}(Y)\ehatl 2{l'}^{r}(\mathbf{Y})\right]\right],\]
\[
I_{2}=\mathbb{E}\left[1_{\{\mathbf{X},\mathbf{Y}\}\in\Psi(l,l')}\gamma_{1}(\mathbf{X})\gamma_{2}(\mathbf{Y})Cov_{\mathbf{X,Y}}\left[\ehatl 1l{}^{s}(\mathbf{X})\ehatl 2l^{q}(\mathbf{X}),\ehatl 1{l'}^{t}(Y)\ehatl 2{l'}^{r}(\mathbf{Y})\right]\right].\]
Combining Eqs.~\ref{eq:cov_inside1} and~\ref{eq:cov_inside2} gives
\begin{eqnarray*}
I_{1} & = & \mathbb{E}\left[1_{\{\mathbf{X},\mathbf{Y}\}\in\Psi(l,l')^{C}}\gamma_{1}(\mathbf{X})\gamma_{2}(\mathbf{Y})O\left(\frac{1}{k(l)^{\frac{s+q}{2}}k(l')^{\frac{r+t}{2}}}\right)\right]\\
 & = & \int\left(O\left(\frac{1}{k(l)^{\frac{s+q}{2}}k(l')^{\frac{r+t}{2}}}\right)\left(\gamma_{1}(x)\gamma_{2}(x)+o(1)\right)\right)\left(\int_{\{x,y\}\in\Psi(l,l')^{C}}dy\right)dx\\
 & = & \int\left(O\left(\frac{1}{k(l)^{\frac{s+q}{2}}k(l')^{\frac{r+t}{2}}}\right)\left(\gamma_{1}(x)\gamma_{2}(x)+o(1)\right)\right)\left(2^{d}\frac{\max(k(l),k(l'))}{M}\right)dx\\
 & = & o\left(\frac{1}{M}\right).\end{eqnarray*}
Similarly, \[
I_{2}=\mathbb{E}\left[1_{\{\mathbf{X},\mathbf{Y}\}\in\Psi(l,l')}\gamma_{1}(\mathbf{X})\gamma_{2}(\mathbf{Y})o\left(\frac{1}{M}\right)\right]=o\left(\frac{1}{M}\right),\]
 and so \begin{equation}
Cov\left[\gamma_{1}(\mathbf{X})\ehatl 1l^{s}(\mathbf{X})\ehatl 2l^{q}(\mathbf{X}),\gamma_{2}(\mathbf{Y})\ehatl 1{l'}^{t}(\mathbf{Y})\ehatl 2{l'}^{r}(\mathbf{Y})\right]=o\left(\frac{1}{M}\right).\label{eq:cov_fourterms}\end{equation}

Assume now that neither $E_{0}$, $E_{1,1}$, nor $E_{1,2}$. If either
$q,r=0$ or $s,t=0$ and the remaining exponents are nonzero, then
the left hand side of Eq.~\ref{eq:cov_fourterms} reduces to Eq.~\ref{eq:cov_ells_final}.
For the other cases, suppose that $s,q=0$ and $t,r\geq2$ as an example.
Then we have that \begin{eqnarray*}
Cov\left[\gamma_{1}(\mathbf{X}),\gamma_{2}(\mathbf{Y})\ehatl 1{l'}^{t}(\mathbf{Y})\ehatl 2{l'}^{r}(\mathbf{Y})\right] & = & \mathbb{E}\left[\gamma_{1}(\mathbf{X})\gamma_{2}(\mathbf{Y})\ehat 1^{t}(\mathbf{Y})\ehat 2^{r}(\mathbf{Y})\right]\\
 &  & -\mathbb{E}\left[\gamma_{1}(\mathbf{X})\right]\mathbb{E}\left[\gamma_{2}(\mathbf{Y})\ehat 1^{t}(\mathbf{Y})\ehat 2^{r}(\mathbf{Y})\right]\\
 & = & \mathbb{E}\left[\gamma_{1}(\mathbf{X})\right]\mathbb{E}\left[\gamma_{2}(\mathbf{Y})\ehat 1^{t}(\mathbf{Y})\ehat 2^{r}(\mathbf{Y})\right]\\
 &  & -\mathbb{E}\left[\gamma_{1}(\mathbf{X})\right]\mathbb{E}\left[\gamma_{2}(\mathbf{Y})\ehat 1^{t}(\mathbf{Y})\ehat 2^{r}(\mathbf{Y})\right]\\
 & = & 0.\end{eqnarray*}
 The same result follows for all other cases.

Finally, applying Eqs.~\ref{eq:cov_ells_final} and \ref{eq:cov_fourterms}
to Eq.~\ref{eq:Fk_q} gives \[
Cov\left[\gamma_{1}(\mathbf{X})\ekl l^{q}(\mathbf{X}),\gamma_{2}(\mathbf{Y})\ekl{l'}^{r}(\mathbf{Y})\right]=o\left(\frac{1}{M}\right)+1_{\{q=1,r=1\}}\left(\frac{1}{M}\right)\times\]
\[
\left(c_{7,1}\left(\frac{\gamma_{1}(x)}{\mathbb{E}_{X}\fhatl 2l(x)},\frac{\gamma_{2}(x)}{\mathbb{E}_{X}\fhatl 2{l'}(x)}\right)+c_{7,2}\left(\frac{-\gamma_{1}(x)\mathbb{E}_{X}\fhatl 1l(x)}{\mathbb{E}_{X}\fhatl 2l(x)},\frac{-\gamma_{2}(x)\mathbb{E}_{X}\fhatl 1{l'}(x)}{\mathbb{E}_{X}\fhatl 2{l'}(x)}\right)\right)\]
\[
=1_{\{q=1,r=1\}}c_{8}\left(\gamma_{1}(x),\gamma_{2}(x)\right)\left(\frac{1}{M}\right)+o\left(\frac{1}{M}\right).\]
 Note that this holds even if $l=l'$.
\end{proof}
The following lemma is required to bound the $\Psi(\mathbf{Z})$ term. 
\begin{lem}
\label{lem:bounded}Assume that U(x) is any arbitrary functional which
satisfies\begin{eqnarray*}
(i) & \mathbb{E}\left[\sup_{L\in(p_{l},p_{u})}\left|U\left(L\frac{\mathbf{p}_{2}}{\mathbf{p}_{1}}\right)\right|\right]=G_{1}<\infty,\\
(ii) & \sup_{L\in\left(\frac{q_{l,1}}{q_{u,2}},\frac{q_{u,1}}{q_{l,2}}\right)}\left|U\left(L\right)\right|\mathcal{C}\left(k_{1}\right)\mathcal{C}\left(k_{2}\right)=G_{2}<\infty,\\
(iii) & \mathbb{E}\left[\sup_{L\in\left(\frac{q_{l,1}}{p_{u,2}},\frac{q_{u,1}}{p_{l,2}}\right)}\left|U\left(L\mathbf{p}_{2}\right)\right|\mathcal{C}\left(k_{1}\right)\right]=G_{3}<\infty,\\
(iv) & \mathbb{E}\left[\sup_{L\in\left(\frac{p_{l,1}}{q_{u,2}},\frac{p_{u,1}}{q_{l,2}}\right)}\left|U\left(\frac{L}{\mathbf{p}_{1}}\right)\right|\mathcal{C}\left(k_{2}\right)\right]=G_{4}<\infty.\end{eqnarray*}
Let $\mathbf{Z}$ be $\mathbf{X}_{i}$ for some fixed $i\in\{1,\dots,N\}$
and $\xi_{\mathbf{Z}}$ be any random variable which almost surely
lies in $(L(\mathbf{Z}),\lhat(\mathbf{Z})).$ Then $\mathbb{E}|U(\xi_{\mathbf{Z}})|<\infty.$ \end{lem}
\begin{proof}
This is a version of Lemma~9 in~\cite{sricharan2013ensemble} modified
to apply to functionals of the likelihood ratio. Because of assumption
$\mathcal{A}.1$, it is sufficient to show that the conditional expectation
$\mathbb{E}\left[|U(\xi_{Z})|\,|\,\mathbf{X}_{1},\dots,\mathbf{X}_{N}\right]<\infty.$ 

First, some properties of $k$-NN density estimators are required.
Let $\mathbf{S}_{k_{i},i}(Z)=\left\{ Y:d(Z,Y)\leq\mathbf{d}_{Z,i}^{(k_{i})}\right\} $
where $\mathbf{d}_{Z,i}^{(k_{i})}$ is the distance to the $k_{i}$th
nearest neighbor of $Z$ from the corresponding set of samples. Then
let $\mathbf{P}_{i}(Z)=\int_{\mathbf{S}_{k_{i},i}(Z)}f_{i}(x)dx$
which has a beta distribution with parameters $k_{i}$ and $M_{i}-k_{i}+1$~\cite{mack1979knn}.
Let $A_{i}(Z)$ be the event that $\mathbf{P}_{i}(Z)<\left(\frac{\sqrt{6}}{k_{i}^{\delta/2}}+1\right)\frac{k_{i}-1}{M_{i}}.$
It has been shown that $Pr\left(A_{i}(Z)^{C}\right)=\Theta\left(\mathcal{C}\left(k_{i}\right)\right)$
and that under $A_{i}(Z)$~\cite{sricharan2012estimation,kumar2012thesis},
\[
\frac{p_{l,i}}{\mathbf{P}_{i}(Z)}<\fhat i(Z)<\frac{p_{u,i}}{\mathbf{P}_{i}(Z)}.\]
It has also been shown that under $A_{i}(Z)^{C}$~\cite{sricharan2012estimation,kumar2012thesis},
\[
q_{l,i}<\fhat i(Z)<q_{u,i}.\]
Let $A(Z)=A_{1}(Z)\cap A_{2}(Z)$ and note that $A_{1}(Z)$ and $A_{2}(Z)$
are independent events. Thus since $\lhat(Z)=\frac{\fhat 1(Z)}{\fhat 2(Z)},$
we have that under $A(Z),$ \[
p_{l}\frac{\mathbf{P}_{2}(Z)}{\mathbf{P}_{1}(Z)}<\lhat(Z)<p_{u}\frac{\mathbf{P}_{2}(Z)}{\mathbf{P}_{1}(Z)}.\]
Now let $Q_{1}(Z)=A_{1}(Z)^{C}\cap A_{2}(Z)^{C},$ $Q_{2}(Z)=A_{1}(Z)^{C}\cap A_{2}(Z),$
and $Q_{3}(Z)=A_{1}(Z)\cap A_{2}(Z)^{C}.$ Then due to independence
and the fact that the $Q_{i}(Z)$s are disjoint, \begin{eqnarray*}
A(Z)^{C} & = & A_{1}(Z)^{C}\cup A_{2}(Z)^{C}=Q_{1}(Z)\cup Q_{2}(Z)\cup Q_{3}(Z),\\
\implies Pr\left(A(Z)^{C}\right) & = & Pr\left(Q_{1}(Z)\right)+Pr\left(Q_{2}(Z)\right)+Pr\left(Q_{3}(Z)\right)\\
 & \leq & \mathcal{C}\left(k_{1}\right)\mathcal{C}\left(k_{2}\right)+\mathcal{C}\left(k_{1}\right)+\mathcal{C}\left(k_{2}\right).\end{eqnarray*}
Then under $Q_{1}(Z)$, $Q_{2}(Z),$ and $Q_{3}(Z)$, respectively,\[
\frac{q_{l,1}}{q_{u,2}}<\lhat(Z)<\frac{q_{u,1}}{q_{l,2}},\]
\[
\frac{q_{l,1}\mathbf{P}_{2}(Z)}{p_{u,2}}<\lhat(Z)<\frac{q_{u,1}\mathbf{P}_{2}(Z)}{p_{l,2}},\]
\[
\frac{p_{l,1}}{\mathbf{P}_{1}(Z)q_{u,2}}<\lhat(Z)<\frac{p_{u,1}}{\mathbf{P}_{1}(Z)q_{l,2}}.\]
Conditioning on $\mathbf{X}_{1},\dots,\mathbf{X}_{N}$ gives \begin{eqnarray*}
\mathbb{E}\left[\left|U(\xi_{Z})\right|\right] & = & \mathbb{E}\left[1_{A(Z)}\left|U(\xi_{Z})\right|\right]+\mathbb{E}\left[1_{Q_{1}(Z)}\left|U(\xi_{Z})\right|\right]+\mathbb{E}\left[1_{Q_{2}(Z)}\left|U(\xi_{Z})\right|\right]+\mathbb{E}\left[1_{Q_{3}(Z)}\left|U(\xi_{Z})\right|\right]\\
 & \leq & Pr(A(Z))\mathbb{E}\left[\sup_{L\in(p_{l},p_{u})}\left|U\left(L\frac{\mathbf{P}_{2}(Z)}{\mathbf{P}_{1}(Z)}\right)\right|\right]+Pr(Q_{1}(Z))\sup_{L\in\left(\frac{q_{l,1}}{q_{u,2}},\frac{q_{u,1}}{q_{l,2}}\right)}\left|U\left(L\right)\right|\\
 &  & +Pr(Q_{1}(Z))\mathbb{E}\left[\sup_{L\in\left(\frac{q_{l,1}}{p_{u,2}},\frac{q_{u,1}}{p_{l,2}}\right)}\left|U\left(L\mathbf{P}_{2}(Z)\right)\right|\right]\\
 &  & +Pr(Q_{1}(Z))\mathbb{E}\left[\sup_{L\in\left(\frac{p_{l,1}}{q_{u,2}},\frac{p_{u,1}}{q_{l,2}}\right)}\left|U\left(\frac{L}{\mathbf{P}_{1}(Z)}\right)\right|\right]\\
 & \leq & \mathbb{E}\left[\sup_{L\in(p_{l},p_{u})}\left|U\left(L\frac{\mathbf{P}_{2}(Z)}{\mathbf{P}_{1}(Z)}\right)\right|\right]+\sup_{L\in\left(\frac{q_{l,1}}{q_{u,2}},\frac{q_{u,1}}{q_{l,2}}\right)}\left|U\left(L\right)\right|\mathcal{C}\left(k_{1}\right)\mathcal{C}\left(k_{2}\right)\\
 &  & +\mathbb{E}\left[\sup_{L\in\left(\frac{q_{l,1}}{p_{u,2}},\frac{q_{u,1}}{p_{l,2}}\right)}\left|U\left(L\mathbf{P}_{2}(Z)\right)\right|\mathcal{C}\left(k_{1}\right)\right]\\
 &  & +\mathbb{E}\left[\sup_{L\in\left(\frac{p_{l,1}}{q_{u,2}},\frac{p_{u,1}}{q_{l,2}}\right)}\left|U\left(\frac{L}{\mathbf{P}_{1}(Z)}\right)\right|\mathcal{C}\left(k_{2}\right)\right]\\
 & = & G_{1}+G_{2}+G_{3}+G_{4}<\infty.\end{eqnarray*}

\end{proof}
The next lemma gives the last result necessary to bound the covariance
of $\mathbf{Y}_{M,1}$ and $\mathbf{Y}_{M,2}$.
\begin{lem}
\label{lem:g_cov}Let $l,l'\in\bar{l}$ be fixed, $M_{1}=M_{2}=M$,
$k(l)=l\sqrt{M}$, and $\lhatl l=\hat{\mathbf{L}}_{k(l),k(l)}$. Let
$\mathbf{X}_{i}$ and $\mathbf{X}_{j}$ be realizations of the density
$f_{2}$ independent of $\fhat 1$ and $\fhat 2$ and independent
of each other when $i\neq j$. Then\[
Cov\left[g\left(\lhatl l(\mathbf{X}_{i})\right),g\left(\lhatl{l'}(\mathbf{X}_{j})\right)\right]\]
 \begin{equation}
=\begin{cases}
\bE\left[\mathbf{p}_{i}^{(l)}\mathbf{p}_{i}^{(l')}\right]+o(1), & i=j\\
c_{8}\left(g'\left(\mathbb{E}_{X}\lhatl l(x)\right),g'\left(\mathbb{E}_{X}\lhatl{l'}(x)\right)\right)\left(\frac{1}{M}\right)+o\left(\frac{1}{M}\right), & i\neq j.\end{cases}\label{eq:cov_ells}\end{equation}
\end{lem}
\begin{proof}
Consider the case where $i=j$. Then applying Lemma~\ref{lem:covariance_ek-1}
to Eq.~\ref{eq:cov_g} gives \[
Cov\left[g\left(\lhatl l(\mathbf{X}_{i})\right),g\left(\lhatl{l'}(\mathbf{X}_{j})\right)\right]=\bE\left[\gtay pil\gtay pi{l'}\right]+o(1).\]
Note that $\bE\left[\gtay pil\gtay pi{l'}\right]=O(1)$ since $\gtay pil=\mathcal{M}\left(g\left(\bE_{\mathbf{X}_{i}}\lhatl l(\mathbf{X}_{i})\right)\right)=\mathcal{M}\left(g\left(L(\mathbf{X}_{i})\right)\right)+o(1)$. 

Now let $i\neq j$. Since $\mathbf{X}_{i}$ and $\mathbf{X}_{j}$
are independent, $\mathbb{E}\left[\gtay pil\left(\gtay pj{l'}+\gtay qj{l'}+\gtay rj{l'}+\gtay sj{l'}\right)\right]=0.$
Applying Lemma~\ref{lem:covariance_ek-1} gives \begin{eqnarray*}
\mathbb{E}\left[\gtay qil\gtay qj{l'}\right] & = & c_{8}\left(g'\left(\mathbb{E}_{X}\lhatl l(x)\right),g'\left(\mathbb{E}_{X}\lhatl{l'}(x)\right)\right)\left(\frac{1}{M}\right)+o\left(\frac{1}{M}\right),\\
\mathbb{E}\left[\gtay qil\gtay rj{l'}\right] & = & o\left(\frac{1}{M}\right),\\
\mathbb{E}\left[\gtay ril\gtay rj{l'}\right] & = & o\left(\frac{1}{M}\right).\end{eqnarray*}
We use Cauchy-Schwarz and Lemma~\ref{lem:ekhat} to get\begin{align*}
\mathbb{E}\left[g'\left(\mathbb{E}_{\mathbf{X}_{i}}\lhatl l\left(\mathbf{X}_{i}\right)\right)\ekl l\left(\mathbf{X}_{i}\right)\Psi\left(\mathbf{X}_{j}\right)\ekl{l'}^{\lambda}\left(\mathbf{X}_{j}\right)\right]\end{align*}
 \begin{eqnarray*}
 & \leq & \sqrt{\mathbb{E}\left[\Psi^{2}\left(\mathbf{X}_{j}\right)\right]\mathbb{E}\left[\left(g'\left(\mathbb{E}_{\mathbf{X}_{i}}\lhatl l\left(\mathbf{X}_{i}\right)\right)\ekl l\left(\mathbf{X}_{i}\right)\right)^{2}\ekl{l'}^{2\lambda}\left(\mathbf{X}_{j}\right)\right]}\\
 & \leq & \sqrt{\mathbb{E}\left[\Psi^{2}\left(\mathbf{X}_{j}\right)\right]\sqrt{\mathbb{E}\left[\left(g'\left(\mathbb{E}_{\mathbf{X}_{i}}\lhatl l\left(\mathbf{X}_{i}\right)\right)\ekl l\left(\mathbf{X}_{i}\right)\right)^{4}\right]\mathbb{E}\left[\ekl{l'}^{4\lambda}\left(\mathbf{X}_{j}\right)\right]}}\\
 & = & \sqrt{\mathbb{E}\left[\Psi^{2}\left(\mathbf{X}_{j}\right)\right]\sqrt{O\left(\frac{1}{k(l)^{2}}\right)O\left(\frac{1}{k(l')^{2\lambda}}\right)}}\\
 & = & \sqrt{\mathbb{E}\left[\Psi^{2}\left(\mathbf{X}_{j}\right)\right]}o\left(\frac{1}{k(l')^{\lambda/2}}\right).\end{eqnarray*}
Lemma~\ref{lem:bounded} and assumption $\left(\mathcal{A}.5\right)$
implies that $\mathbb{E}\left[\Psi^{2}\left(\mathbf{X}_{j}\right)\right]=O(1)$
and from assumption $(\mathcal{A}.3),$ $o\left(\frac{1}{k(l')^{\lambda/2}}\right)=o\left(\frac{1}{M}\right).$
This implies that $\mathbb{E}\left[\gtay qil\gtay sj{l'}\right]=o\left(\frac{1}{M}\right).$
Similarly, $\mathbb{E}\left[\gtay ril\gtay sj{l'}\right]=o\left(\frac{1}{M}\right)$
and $\mathbb{E}\left[\gtay sil\gtay sj{l'}\right]=o\left(\frac{1}{M}\right).$
Combining these results with Eq.~\ref{eq:cov_g} completes the proof.
\end{proof}
Applying Lemma~\ref{lem:g_cov} to Eqs.~\ref{eq:y_num} and \ref{eq:y_den}
shows that $Cov(\mathbf{Y}_{M,i},\mathbf{Y}_{M,j})=O\left(\frac{1}{M}\right)$.

For the covariance of $\mathbf{Y}_{M,i}^{2}$ and $\mathbf{Y}_{M,j}^{2}$,
we only need to consider the numerator since we previously showed
that $\var\left[\sum_{l\in\bar{l}}w(l)g\left(\lhatl l(\mathbf{X}_{i})\right)\right]=O(1)+o(1)$.
Assume WLOG that $i=1$ and $j=2$ and let $h_{l}=\bE\left[g\left(\lhatl l(\mathbf{X}_{i})\right)\right]$.
The numerator of the covariance is then\[
\sum_{l\in\bar{l}}\sum_{l'\in\bar{l}}\sum_{j\in\bar{l}}\sum_{j'\in\bar{l}}Cov\left[\left(g\left(\lhatl l(\mathbf{X}_{1})\right)-h_{l}\right)\left(g\left(\lhatl{l'}(\mathbf{X}_{1})\right)-h_{l'}\right),\right.\]
\[
\left.\left(g\left(\lhatl j(\mathbf{X}_{2})\right)-h_{j}\right)\left(g\left(\lhatl{j'}(\mathbf{X}_{2})\right)-h_{j'}\right)\right]\]
\[
=\sum_{l\in\bar{l}}\sum_{l'\in\bar{l}}\sum_{j\in\bar{l}}\sum_{j'\in\bar{l}}Cov\left[\left(\gtay p1l+\gtay q1l+\gtay r1l+\gtay s1l\right)\left(\gtay p1{l'}+\gtay q1{l'}+\gtay r1{l'}+\gtay s1{l'}\right),\right.\]
\[
\left.\left(\gtay p2j+\gtay q2j+\gtay r2j+\gtay s2j\right)\left(\gtay p2{j'}+\gtay q2{j'}+\gtay r2{j'}+\gtay s2{j'}\right)\right].\]
 Let $d_{l}(x)=\left(g(x)-h_{l}\right)^{2}$. Then for the case where
$l=l'$ and $j=j'$, we have \[
Cov\left[d_{l}\left(\lhatl l(\mathbf{X}_{1})\right),d_{j}\left(\lhatl j(\mathbf{X}_{2})\right)\right]=O\left(\frac{1}{M}\right).\]
This follows from Lemma~\ref{lem:g_cov}. 

For the general case, note that due to the independence of $\mathbf{X}_{1}$
and $\mathbf{X}_{2}$, \begin{eqnarray*}
Cov\left[\gtay p1l\gtay p1{l'},\left(g\left(\lhatl j(\mathbf{X}_{2})\right)-h_{j}\right)\left(g\left(\lhatl{j'}(\mathbf{X}_{2})\right)-h_{j'}\right)\right] & = & 0,\\
Cov\left[\left(g\left(\lhatl l(\mathbf{X}_{1})\right)-h_{l}\right)\left(g\left(\lhatl{l'}(\mathbf{X}_{1})\right)-h_{l'}\right),\gtay p2j\gtay p2{j'}\right] & = & 0.\end{eqnarray*}
To bound the remaining terms, we require the following Lemma:
\begin{lem}
\label{lem:cov_Fk4}Let $\gamma_{1}(x),$ $\gamma_{2}(x)$ be arbitrary
functions with $1$ partial derivative wrt $x$ and $\sup_{x}|\gamma_{i}(x)|<\infty,\, i=1,\,2.$
Let $l,l',j,j'\in\bar{l}$ be fixed, $M_{1}=M_{2}=M$, $k(l)=l\sqrt{M}$.
Let $\mathbf{X}$ and $\mathbf{Y}$ be realizations of the density
$f_{2}$ independent of $\fhatl il$, $\fhatl i{l'}$, $\fhatl ij$,
and $\fhatl i{j'}$, $i=1,\,2$. If $q,r,s,t\geq0$ and the cases
$\left\{ t=0,r=0\right\} $ or $\{q=0,s=0\}$ do not hold, then \[
Cov\left[\gamma_{1}(\mathbf{X})\ekl l^{s}(\mathbf{X})\ekl{l'}^{q}(\mathbf{X}),\gamma_{2}(\mathbf{Y})\ekl j^{t}(\mathbf{Y})\ekl{j'}^{r}(\mathbf{Y})\right]=O\left(\frac{1}{M}\right).\]
\end{lem}
\begin{proof}
Under certain conditions, by Cauchy-Schwarz and Eq.~\ref{eq:moment3}
we have \[
Cov\left[\gamma_{1}(\mathbf{X})\ehatl 1l^{s}(\mathbf{X})\ehatl 2l^{q}(\mathbf{X})\ehatl 1{l'}^{s'}(\mathbf{X})\ehatl 2{l'}^{q'}(\mathbf{X}),\right.\]
\[
\left.\gamma_{2}(\mathbf{Y})\ehatl 1j^{t}(\mathbf{Y})\ehatl 2j^{r}(\mathbf{Y})\ehatl 1{j'}^{t'}(\mathbf{Y})\ehatl 1{j'}^{r'}(\mathbf{Y})\right]\]
\[
\leq\bE\left[\gamma_{1}(\mathbf{X})\gamma_{2}(\mathbf{Y})\sqrt{\var_{\mathbf{X}}\left[\ehatl 1l^{s}(\mathbf{X})\ehatl 2l^{q}(\mathbf{X})\ehatl 1{l'}^{s'}(\mathbf{X})\ehatl 2{l'}^{q'}(\mathbf{X})\right]}\right.\]
\[
\left.\times\sqrt{\var_{\mathbf{Y}}\left[\ehatl 1j^{t}(\mathbf{Y})\ehatl 2j^{r}(\mathbf{Y})\ehatl 1{j'}^{t'}(\mathbf{Y})\ehatl 1{j'}^{r'}(\mathbf{Y})\right]}\right]\]
\begin{equation}
=O\left(\frac{1}{k(l)^{\frac{s+q}{2}}k(l')^{\frac{s'+q'}{2}}k(j)^{\frac{t+r}{2}}k(j')^{\frac{t'+r'}{2}}}\right)=O\left(\frac{1}{M^{\frac{s+q+s'+q'+t+r+t'+r'}{4}}}\right).\label{eq:cov_ek_lots}\end{equation}
Note that the exponents $q,s,r,t$ are not the same as in the statement
of the lemma. The conditions under which this expression holds are
as follows: (1) There must be at least one positive exponent on both
sides of the arguments in the covariance. (2) $\{s+s'+t+t'\neq1\}\cap\{q+q'+r+r'\neq1\}$.
If neither case holds, this reduces to Eq.~\ref{eq:cov_ells_final}.
If only one holds, then the covariance is zero. 

Note that if $s+q+s'+q'+t+r+t'+r'\geq4$, Eq.~\ref{eq:cov_ek_lots}
becomes $O\left(\frac{1}{M}\right)$. Now consider the case where
$\{\{s+s'+t+t'=3\}\cap\{s,s',t,t'\leq1\}\cap\{q,q',r,r'=0\}\}\cup\{\{q+q'+r+r'=3\}\cap\{q,q',r,r'\leq1\}\cap\{s,s',t,t'=0\}\}$.
Assume WLOG that $s,s',t=1$. Then Eq.~\ref{eq:cov_ek_lots} becomes
$O\left(\frac{1}{M^{\frac{3}{4}}}\right)$ which does not decay fast
enough to use Lemma~\ref{lem:clt_covariance-1}. However, we can
use the fact that $k(l)=O(k(l'))$ to obtain a bound of $O\left(\frac{1}{M}\right)$.
By Markov's inequality and Eqs.~\ref{eq:moment} and~\ref{eq:cross_moment},
for fixed $\nu>0$, \begin{eqnarray*}
Pr\left(\left|\ehatl 1l(\mathbf{X})-\ehatl 1{l'}(\mathbf{X})\right|>\nu\right) & \leq & \frac{\bE\left[\left(\ehatl 1l(\mathbf{X})-\ehatl 1{l'}(\mathbf{X})\right)^{4}\right]}{\nu^{4}}\\
 & = & O\left(\frac{1}{M}\right).\end{eqnarray*}
Let $H$ be the event that $\left|\ehatl 1l(\mathbf{X})-\ehatl 1{l'}(\mathbf{X})\right|\leq1$.
This gives \begin{eqnarray}
 &  & Cov\left[\gamma_{1}(\mathbf{X})\ehatl 1l(\mathbf{X})\ehatl 1{l'}(\mathbf{X}),\gamma_{2}(\mathbf{Y})\ehatl 1j(\mathbf{Y})\right]\nonumber \\
 & = & \bE\left[\gamma_{1}(\mathbf{X})\gamma_{2}(\mathbf{Y})\ehatl 1l(\mathbf{X})\ehatl 1{l'}(\mathbf{X})\ehatl 1j(\mathbf{Y})\right]\nonumber \\
 & = & \bE\left[\mathbf{1}_{H}\gamma_{1}(\mathbf{X})\gamma_{2}(\mathbf{Y})\ehatl 1l(\mathbf{X})\ehatl 1{l'}(\mathbf{X})\ehatl 1j(\mathbf{Y})\right]\nonumber \\
 &  & +\bE\left[\mathbf{1}_{H^{C}}\gamma_{1}(\mathbf{X})\gamma_{2}(\mathbf{Y})\ehatl 1l(\mathbf{X})\ehatl 1{l'}(\mathbf{X})\ehatl 1j(\mathbf{Y})\right]\nonumber \\
 & \leq & \bE\left[\mathbf{1}_{H}\gamma_{1}(\mathbf{X})\gamma_{2}(\mathbf{Y})\ehatl 1l^{2}(\mathbf{X})\ehatl 1j(\mathbf{Y})\right]+\bE\left[\mathbf{1}_{H}\gamma_{1}(\mathbf{X})\gamma_{2}(\mathbf{Y})\ehatl 1l(\mathbf{X})\ehatl 1j(\mathbf{Y})\right]\nonumber \\
 &  & +\bE\left[\mathbf{1}_{H^{C}}\gamma_{1}(\mathbf{X})\gamma_{2}(\mathbf{Y})\ehatl 1l(\mathbf{X})\ehatl 1{l'}(\mathbf{X})\ehatl 1j(\mathbf{Y})\right]\nonumber \\
 & = & O\left(\frac{1}{M}\right).\label{eq:cov_ells3}\end{eqnarray}
The final step for the first two terms comes from Eq.~\ref{eq:cov_ells_final}.
The final step for the third term comes from the fact that $Pr(H^{C})=O\left(\frac{1}{M}\right)$
and the fact that $\bE\left[\gamma_{1}(\mathbf{X})\gamma_{2}(\mathbf{Y})\ehatl 1l(\mathbf{X})\ehatl 1{l'}(\mathbf{X})\ehatl 1j(\mathbf{Y})\right]=o(1)$
by Eq.~\ref{eq:cov_ek_lots}.\textbf{\emph{ }}Applying Eqs.~\ref{eq:cov_ells_final},
\ref{eq:cov_fourterms}, \ref{eq:cov_ek_lots}, and \ref{eq:cov_ells3}
to Eq.~\ref{eq:Fk_q} completes the proof.
\end{proof}
From Lemma~\ref{lem:cov_Fk4}, it is clear that \begin{eqnarray*}
Cov\left[\left(\gtay p1l+\gtay q1l+\gtay r1l+\gtay s1l\right)\left(\gtay p1{l'}+\gtay q1{l'}+\gtay r1{l'}+\gtay s1{l'}\right),\right.\\
\left.\left(\gtay p2j+\gtay q2j+\gtay r2j+\gtay s2j\right)\left(\gtay p2{j'}+\gtay q2{j'}+\gtay r2{j'}+\gtay s2{j'}\right)\right] & = & O\left(\frac{1}{M}\right)\\
\implies Cov\left[\mathbf{Y}_{M,i}^{2},\mathbf{Y}_{M,j}^{2}\right] & = & O\left(\frac{1}{M}\right).\end{eqnarray*}
Then by Lemma~\ref{lem:clt_covariance-1}, $\mathbf{S}_{N,M}=\frac{\hat{\mathbf{G}}_{w}-\bE\left[\hat{\mathbf{G}}_{w}\right]}{\sqrt{\var\left[\hat{\mathbf{G}}_{w}\right]}}$
converges in distribution to a standard normal random variable.

\small

\renewcommand{\refname}{\normalsize References}

\bibliographystyle{ieeetr}
\bibliography{divergence_isit2014}

\end{document}